\documentclass[aps,prd,twocolumn,superscriptaddress,nofootinbib,longbibliography]{revtex4-2}
\pdfoutput=1

\usepackage{amsmath,amssymb}
\usepackage{amsthm}
\usepackage[utf8]{inputenc}
\usepackage[english]{babel}
\usepackage{graphicx}
\usepackage{dcolumn}
\usepackage{bm}
\usepackage{braket}
\usepackage{mathtools}
\usepackage{url}
\usepackage{soul}
\usepackage{array}
\usepackage[colorlinks,citecolor=blue,linkcolor=blue,urlcolor=blue, breaklinks=true]{hyperref}
\usepackage{float}
\usepackage{orcidlink}
\makeatletter
\AtBeginDocument{\let\LS@rot\@undefined}
\makeatother

\newcommand{\SN}{\mathrm{S}_N\xspace}
\renewcommand{\S}{\mathrm{S}\xspace}
\renewcommand{\H}{\mathcal{H}\xspace}

\renewcommand{\S}{\mathrm{S}\xspace}
\renewcommand{\H}{\mathcal{H}\xspace}

\newcommand{\NN}{{\mathbb N}\xspace}

\newcommand{\kin}{\mathrm{kin}\xspace}

\DeclareMathOperator{\Inv}{Inv}

\newcommand{\matrixel}[3]{{\langle #1 | #2 | #3  \rangle}}

\usepackage{xcolor}
\usepackage{xspace}
\usepackage{ifthen}
\graphicspath{{/}}

\newcommand{\showcomments}{true}

\newcommand{\andrea}[1]%
{\ifthenelse{\equal{\showcomments}{true}}%
{{\color{orange}{\small \textbf{A:} #1}}}{\xspace}}%

\newcommand{\marios}[1]%
{\ifthenelse{\equal{\showcomments}{true}}%
{{\color{blue}{\small \textbf{M:} #1}}}{\xspace}}%

\newcommand{\eugenio}[1]%
{\ifthenelse{\equal{\showcomments}{true}}%
{{\color{purple}{\small \textbf{Eu:} #1}}}{\xspace}}%

\newcommand{\emil}[1]%
{\ifthenelse{\equal{\showcomments}{true}}%
{{\color{red}{\small \textbf{Em:} #1}}}{\xspace}}%

\newtheorem*{theorem*}{Theorem}

\newtheorem*{definition*}{Definition}

\begin{document}

\begin{abstract}

\noindent We show that the problem of observables can be fully resolved for background independent theories defined on graphs, through the explicit construction of complete observables. The appropriate analogue of coordinate independence is argued to be the invariance under changes of graph labels, a kind of permutation invariance. Invariants are formed by group averaging and they probe the entire graph---they are global. Strikingly, sets of complete observables can be constructed so that each of the invariants comprising them seeks a connected subgraph structure---local correlations. Geometrical information is fully encoded through this subtle interplay of global and local graph notions, a behavior we term glocal.  This provides physically meaningful complete sets of observables for discrete general relativity, and a permutation invariant reformulation of the spin networks state space of loop quantum gravity. Our analysis reveals an important new aspect of the problem of observables, demonstrating a deep connection between the theory of spacetime and computational complexity theory: the construction of a complete set of observables for discrete spacetime theories is computationally costly, as it corresponds to solving a graph isomorphism problem.

\end{abstract}

\author{Emil Broukal\,\orcidlink{0009-0000-6488-9927} }
\affiliation{Institute for Quantum Optics and Quantum Information (IQOQI) Vienna, Austrian Academy of Sciences, Boltzmanngasse 3, A-1090 Vienna, Austria}
\affiliation{Basic Research Community for Physics e.V., Mariannenstraße 89, Leipzig, Germany}

 \author{Andrea {Di Biagio}\,\orcidlink{0000-0001-9646-8457}}
 \affiliation{Institute for Quantum Optics and Quantum Information (IQOQI) Vienna, Austrian Academy of Sciences, Boltzmanngasse 3, A-1090 Vienna, Austria}
 \affiliation{Basic Research Community for Physics e.V., Mariannenstraße 89, Leipzig, Germany}

\author{Eugenio Bianchi\,\orcidlink{0000-0001-7847-9929}}
 \affiliation{Institute for Gravitation and the Cosmos, The Pennsylvania State University, University Park, Pennsylvania 16802, USA}
 \affiliation{Department of Physics, The Pennsylvania State University, University Park, Pennsylvania 16802, USA}

 \author{Marios Christodoulou\,\orcidlink{0000-0001-6818-2478}}
 \affiliation{Institute for Quantum Optics and Quantum Information (IQOQI) Vienna, Austrian Academy of Sciences, Boltzmanngasse 3, A-1090 Vienna, Austria}

\date{\today}

\title{Observables are glocal}
\maketitle

 Since mathematics is a language, there will be redundancy of description: the same thing can be said in many different ways. Much has been discussed of what are the mathematical objects within a physical theory we would like to understand as `true things'. It seems natural to demand that they remain the same irrespectively of the descriptive convention we choose. 

This idea was a guiding insight in the search to find the theory of general relativity \cite{sep-spacetime-holearg}. Defining solutions to be the stationary points of a diffeomorphism-invariant action gives a definition of the theory invariant under changes of coordinates. Then, it seems immediate that invariant functions on the space of solutions, traditionally called the \emph{observables}, should be designated as the `true things'. This leads to a host of issues, identified a long time ago \cite{Old0,Old1,Old2,BG,Obse,Old3,stachel_hole_2014,Iftime:2005gs}, known as `the problem of observables', as observables are badly defined  in the context of differential geometry \cite{Torre, Neg,smolin1992learnstudynonperturbativequantum,Dittrich:2015vfa}. In fact, recently, observables that could tell apart any two different spacetimes that are solutions of general relativity---complete observables---were shown to be (Borel) non-constructible \cite{Panagiotopoulos_2023}.

Of course, invariant functions on the space of solutions are not the only way to work with a diffeomorphism invariant theory. General relativity describes a multitude of intricate physics to minute detail, despite its problem of observables. Given the difficulties, it may seem there is nothing useful to be learned by investigating these invariants. Indeed, in many contexts it will be more useful to work with gauge-fixed quantities and many alternative ways to define `observable quantities' in general relativity have been discussed 
\cite{New0,New1,New2,New4,New5,Goeller:2022rsx}.

However, the difficulties seem to be issues native to the continuum. What if we consider theories where the substratum is not a manifold but a discrete structure, like a graph?  That is, a discrete topological structure with weights assigned to its edges and nodes, similar to fields that live on top of the manifold. Does the `problem of observables' then vanish? How might this relate to the fact that the invariants of general relativity are global, in the sense that they must be functionals of the spacetime metric on the entire manifold \cite{Old2, New3, Donnelly_2016, Donnelly:2016rvo, smolin1992learnstudynonperturbativequantum, adlam2024observerelationalobservables, rovelli2022philosophicalfoundationsloopquantum, New5}? The implication is that \emph{all local information can be completely encoded into global objects}. How precisely can (sets of) global objects \emph{completely} capture local information?  

With these questions in mind, we pose and study the problem of observables in the context of background-independent theories on graphs. We first argue that the analogue of invariance under changes of coordinates for discrete geometries is the invariance under permutations of node labels (Sections \ref{ref:LabelIndependence} and \ref{sec:firstSign}). The induced action preserves the graph structure, the adjacency relations between the weights, leaving physical information invariant. 

Observables are then provided by algebraic graph invariants, \emph{global} objects. Complete sets of observables can be constructed, so that each is encoding information about a type of connected correlation, \emph{local} features (Sections \ref{sec:unlabelledgraphs}, \ref{sec:InvGraphPols} and \ref{sec:glocal}). Therefore, \emph{all} geometrical information on a weighted graph can be captured through objects that globally encode local information---which we will call \emph{glocal} observables---and from which all other observables can be constructed.

 We apply the formalism to discrete general relativity, and to the spin networks state space of loop quantum gravity (Section \ref{regge calc}), demonstrating that complete sets of glocal observables can be constructed using invariant graph polynomials. These constructions demonstrate a novel and deep connection between general relativity and computational complexity theory: the problem of observables posed in a discrete setting concerns the same set of issues as deciding whether two graphs are the same up to isomorphism.

\section{Label independence}
\label{ref:LabelIndependence}
We begin by making the case for a formal analogy between label independence when working with graphs and the invariance under changes of coordinates when working with continuous geometries.

\subsection{Why diffeomorphisms?}

In what sense do diffeomoprhisms correspond to \emph{all} the coordinate changes? If we allow \emph{more} general kinds of coordinate systems, does that imply there are \emph{fewer} physical observables? Allowing for a \emph{larger} group, for instance bijections that are two times differentiable or even piecewise distributional \cite{Steinbauer:2006qi,Bahr:2007tm}, would result in a \emph{smaller} set of invariants compared to allowing for only smooth maps, because the latter are a subgroup of the former. 

Diffeomorphisms can be understood as a subgroup of a much larger group: the bijective maps on the manifold. By definition, these are the permutations on the set of points of the manifold. Consider some manifold with coordinates such that a unique set of reals is assigned to every point. Application of an arbitrary bijective map would again result to each point being mapped to a unique set of reals, a `new choice of labels'. Of course, an arbitrary assignment of labels to points is not a coordinate system and is not going to be useful in any sense. It would break not just differentiability but also measurability, taking us firmly out of the field of applicability of differential geometry.

Nevertheless, smoothness, and even measurability of the mathematical object we use to model physics, are not in any clear sense a priori requirements of a physical principle. They are certainly very useful but we can ask: are they just  a mathematical convenience, technical prerequisites for using the apparatus of differential geometry? Since physically meaningful quantities \emph{cannot} depend on an arbitrary choice of what label we wish to attach to each point of the manifold, should we then be taking the `true things' to be, in principle, the invariants on the space of metrics under all possible labelings generated by arbitrary bijections of the manifold to itself?  

This question is probably impossible to precisely formulate in the continuum. However, we can ask it in the case of discrete physics. Considering ‘all possible labelings’ on a discrete structure is most natural.

 \subsection{The discrete}

Famously, Lagrange showed that in the infinite limit of a chain of masses coupled with springs we arrive at the wave equation: one differential equation rather than a vast number of difference equations \cite{Lagrange}. This demonstrates that while it can be very convenient, physics need not be considered founded on a continuum substrate: the continuum may be a convenient approximation of a discrete world, rather than vice versa. Continuum spacetime is suspected to be emergent from underlying discrete, combinatorial physics \cite{Penrose,Rovelli:1995ac,linnemann2018hintsa,Oriti_2021}. When one works with a discrete theory, the natural symmetries of the mathematics at hand will not relate to notions such as smoothness, native to the continuum. 

Of course, like Lagrange, we may want to do things in a discrete setting \emph{so that} a continuum limit can be arrived at. While in many contexts this can be the goal, formulating a discrete theory having in mind a continuum limit can come at the cost of losing sight of a more clear picture.

In the continuum, the mathematical substratum is typically a manifold. In the discrete, a mathematical substratum with a notion of neighborhood is a graph. We will be giving labels from the naturals to the nodes of a graph. These are the allowed descriptive conventions. Then `true things' should be invariant under \emph{any} such arbitrary choice.  
 
\subsection{Physics without coordinates?}

The above resonates with Regge's eminent work \cite{Regge:1961px}, where it was suggested that general relativity can be defined as a theory of discrete geometry. In Regge calculus, see Section \ref{sec:Regge}, the spacetime metric is encoded in lengths assigned to the edges of a simplicial complex and deficit angles correspond to curvature. 

Regge titled his seminal paper `General relativity without coordinates'. In a sense, this misses the point. Certainly, Regge's idea did not involve \emph{continuous} coordinates. The theory is defined natively on the simplicial complex, which is not embedded in an ambient manifold. The simplicial complex in Regge's version of general relativity \emph{replaces} the manifold. Then, it seems like there are no coordinates in Regge calculus. However, the explicit construction of a generic graph still requires the use of a sort of coordinates, namely a labeling: we begin by setting out the convention `let there be nodes $1,2,3 \dots$'. 

Changes of labels are well recognized as a discrete remnant of continuous changes of coordinates in the program of causal sets \cite{Rideout:1999ub}, where the `allowed' changes of labels are adapted to a partial order. But, we should be able to label points freely: what difference does it make whether we call a point `the point $6$' or `the point $42$'? Changing a labeling convention cannot change the physics. Surely, mathematical points should be taken as indistinguishable.

 In the discrete, the principle analogous to diffeomorphism invariance should in some sense be  the invariance under all changes of labels \cite{Arrighi:2020xnd}. The philosophical case for demanding permutation invariance in theories of discrete spacetime geometry is made in \cite{stachel_hole_2014,Iftime:2005gs}, nicely summarized in \cite{french2003understanding} as 
\begin{quote}
    $\dots$ the status of permutation invariance, from this perspective, is that of one of the fundamental symmetry principles which effectively binds the ‘web of relations’ constituting the structure of the world $[\dots]$, any theory that demands the complete indistinguishability of its fundamental objects requires invariance under the full permutation group for discrete symmetries or the diffeomorphism group for continuous symmetries.
\end{quote} 
  As we will see, these ideas can be implemented concretely. The global objects that capture all local information on a graph in a way that is invariant under arbitrary changes of labels can be systematically constructed.

\section{First sign of glocality}
\label{sec:firstSign}
Before working in full generality, in this Section we work out a few simple examples. We illustrate purely global observables as permutation-invariant functions and show that they can not uniquely determine the state of a graph. Then, we give a first example of a non-trivial (not purely global) glocal observable.

\subsection{The empty graph}
\label{sec:emptyGraph}

Observables with a purely global character suffice to fully characterize the trivial case of the empty graph, a set of points with no adjacency relations---no locality. Take $N$ variables assigned on the nodes through a weighting $g$, a function that given an arbitrary labeling of the nodes assigns to each node $i$ the weight $g_{ii}$. To encode the weight information independently of labels, the $N$ node weights can be traded for the $N$ variables
 \begin{equation}
 \label{eq:powersumNodes}
O^V_K[g] = \sum^N_{i=1} g_{ii}^K 
 \end{equation}
 with $K=1,\ldots, N$. As functions of the node weights, they are invariant under permutations $\sigma$ of the node labels. That is, they satisfy
 \begin{equation}
 \label{eq:globalNodes}
    O^{V}_K(g_{11},\dots, g_{NN}) = O^{V}_K(g_{\sigma(1)\sigma(1)},\dots, g_{\sigma(N)\sigma(N)}),
\end{equation}
for all $\sigma\in\SN$. The $O^V_K$ for $K=1,\dots, N$ form a minimal generating set of all polynomials invariant under this action of permutations, which in turn are dense in the space of all functions that transform as in \eqref{eq:globalNodes}. The set of $N$ values $g_{ii}$ are uniquely reconstructed up to labeling from the $N$ invariants $O^V_K$ by solving for them.\footnote{\relpenalty=0 \binoppenalty=0 For example, take $N=2$. Then $\{g_{11},g_{22}\}=\{\frac{1}{2}(O^V_1-\sqrt{2O^V_2-(O^V_1)^2}),\frac{1}{2}(O^V_1+\sqrt{2O^V_2-(O^V_1)^2})\}$, so one recovers the variables $g_{11},g_{22}$ up to their labels.}

Imagine that the weighting $g$ is a configuration of a sort of `discrete field' living on the empty graph. There is no connectivity, no locality. Then, the $O^V_K$ provide a complete set of observables. Observables of this kind, and any observable that can be constructed from them, capture only purely global information. 

\subsection{The triangle}
Another trivial case where only global information needs to be captured is the geometry of a triangle, a graph with three nodes weighted with lengths on the edges. A triangle $t$ is \emph{fully} defined\footnote{The edge lengths of a Euclidean triangle satisfy triangle inequalities, while those of a Lorentzian triangle do not.} by the set of its edge lengths $\{a,b,c\}$. Two triangles can be distinguished by observables  
  \begin{equation}
 \label{eq:powersumEdges}
O^E_K[t]  = \sum_{e\in E} g_{e}^K
 = a^K+b^K+c^K
 \end{equation}
for $K=1,2,3$. These are of similar form to \eqref{eq:powersumNodes}, where now instead of summing over all node labels the sum is over all the edge labels $E=\{12,13,23\}$. $O^E_1,O^E_2,O^E_3$ form a set of complete observables on the space of triangles, which is also a minimal generating set of all algebraic invariants on that space:
the values of these three polynomials uniquely specify the set $\{a,b,c\}$. For any two non congruent triangles, at least one of them will yield a different number. 

\subsection{Purely global information is not enough}
\label{sec:GlobalToLocal}Beyond the trivial cases of an empty graph and the geometry of a triangle, in addition to the purely global observables, we will also need observables that can capture local information. Indeed, most information on an arbitrary graph is \emph{not} purely global, it concerns adjacency. The observables $O^V_K$ and $O^E_K$ discussed previously---and all observables constructed from them by algebraic and analytic operations--- do not capture any local information. These invariants, which we have called \emph{purely global observables}, are of a quite particular form. We can define them as those that satisfy, for any two fixed indices $i,j$,
\begin{equation}
\label{eq:PurelyGlobalObservables}
    O\big[g\big] = f\bigg(\sum_{\sigma \in \S_N} \phi(g_{\sigma(i)\sigma(j)})\bigg),
\end{equation}
for some functions $f,\phi$ (note that the left hand side does not depend on $i,j$, since it is an invariant). Here $g_{ij}$ are node weights if $i=j$ or edge weights if $i\neq j$, with $\S_N$ the group of permutations on $N$ elements. Remarkably, for any \emph{unordered} set of $N$ variables, it can be proven that all $\S_N$ invariants \emph{must be} of this form \cite{zaheer2018deepsets}. 

Most of the information of a general discrete geometry is not purely global, it is \emph{adjacency information---local correlations}. To be convinced that purely global observables cannot encode any of the graph connectivity we can take all node weights to be 1 and all edge weights to be 1 or 0, imagining that a 0 edge weight `removes' an edge. Then,  $O^V_K=|V|$ and $O^E_K=|E|$ count the number of nodes and the number of edges. Accordingly, the observables in \eqref{eq:PurelyGlobalObservables} will be a function of $|V|$ if $i=j$ and a function of $|E|$ if $i \neq j$. 

The action of permutations on the node labels naturally extends to an action on the edge labels that preserves local, relational information in the weights---what weight is next to which. As a function of the set of weights $g_{ij}$, a general invariant or \emph{observable} satisfies
\begin{equation}
\label{eq:Observables}
    O\big[g_{ij}\big] = O\big[g_{\sigma(i)\sigma(j)}\big]. 
\end{equation}
 That is, observables are invariant under the action of permutations on $N$ elements acting on the labels of $\frac12 N(N+1)$ variables, namely the $N$ node weights and $\frac12 N (N-1)$ edge weights. This far-from-trivial action of permutations yields an intricate structure of invariants. Finite sets of such invariants encode the information of a finite weighted graph. A taste of how this is done can be seen in the case of the tetrahedron.

\subsection{The tetrahedron: glocal information}
\label{ref:tet}
Unlike the triangle, a tetrahedron is not uniquely specified by its set of edge lengths alone. The geometry of a tetrahedron is defined by its six edge lengths \emph{and the adjacency relations between the edges}.\footnote{Six positive reals assigned on the edges of the complete graph with four nodes, \emph{along with their adjacency information}, yield either a Euclidean or a Lorentzian tetrahedron.} Consider the two tetrahedra $T_{ab}$ and $T_{ba}$ depicted in Figure \ref{fig:two tetrahedra}. These are clearly two different geometries. They have three equal edge lengths $a$ and three equal edge lengths $b$, but in one case there is a triangle with all edges $a$ and in the other there is a triangle with all edges $b$. The multiset of six edge lengths $[a,a,a,b,b,b]$ does not distinguish $T_{ab}$ or $T_{ba}$. 
\begin{figure}[H]
    \centering
    \includegraphics[width=0.9\linewidth]{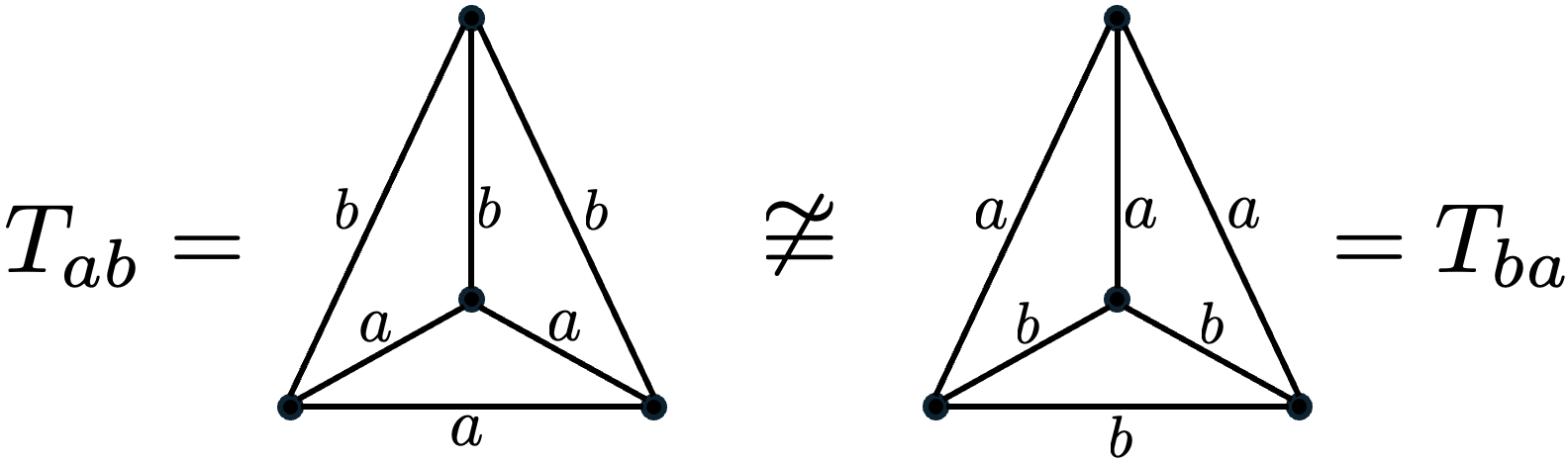}
    \caption{Two tetrahedra with the same multiset of edge lengths $\{a,a,a,b,b,b\}$. If $a \neq b$, they are not isomorphic.}
    \label{fig:two tetrahedra}
\end{figure}
Consider the observable $O$ defined through the polynomial
\begin{equation} 
\begin{aligned}
&O=x_{[1,2]}x_{[1,3]}x_{[1,4]}+x_{[1,2]}x_{[2,3]}x_{[2,4]}\\
&+x_{[1,3]}x_{[2,3]}x_{[3,4]}+x_{[1,4]}x_{[2,4]}x_{[3,4]}\label{tetra obs}  
\end{aligned}
\end{equation}
where $x_{[i,j]}$ are functions returning the weights of edges between nodes labeled $i$ and $j$. Each term is a monomial and the sum of these monomials yields an invariant under changes of labels,  it satisfies \eqref{eq:Observables}.   For $T_{ab}$ and $T_{ba}$, we have
\begin{align}
    O[T_{ab}]=(3a^2b+b^3), ~~O[T_{ba}]=(3ab^2+a^3).
\end{align} 
 For arbitrary $a\neq b$, these two numbers will disagree. Therefore, this observable witnesses that $T_{ab}$ and $T_{ba}$ are two different geometries. 

Anticipating the notation of Section \ref{sec:glocal}, the observable $O$ can be expressed in a graphical calculus as
\begin{equation}
    \vcenter{\hbox{\includegraphics[width=0.8\linewidth]{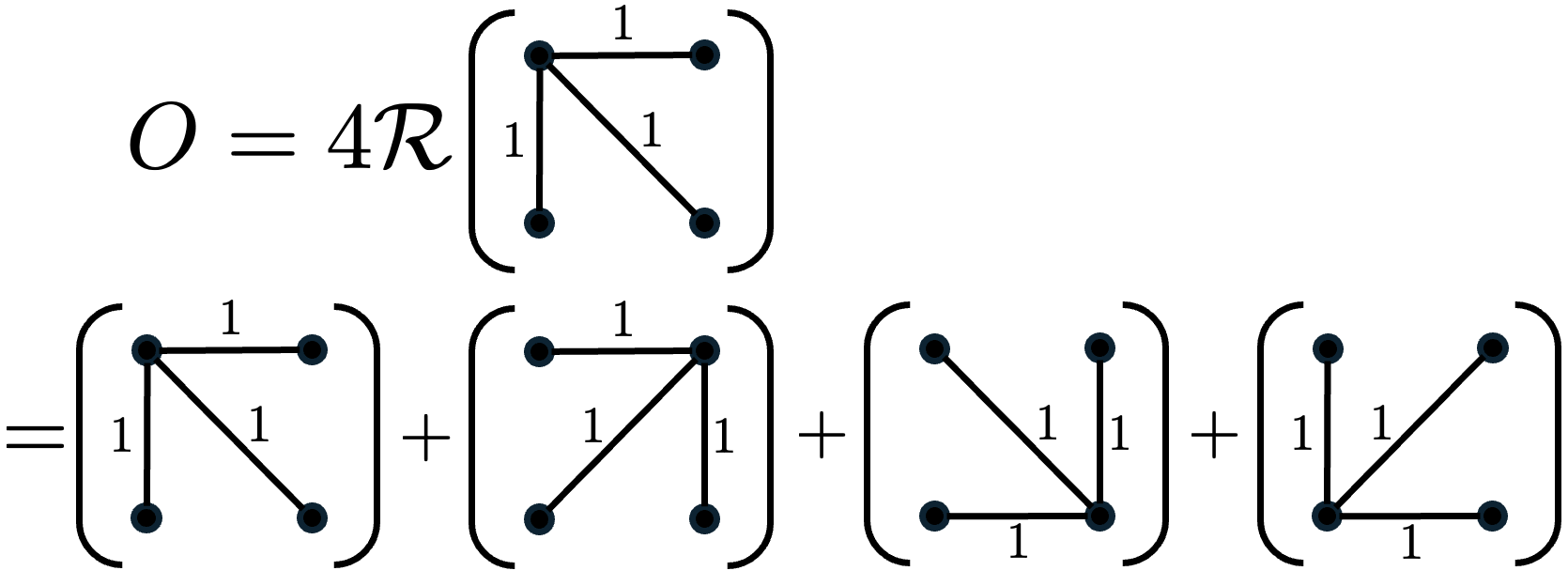}}}.
\end{equation}
The operation $\mathcal{R}$ is a group averaging over changes of labels, defined later on. In words, $\mathcal{R}$ finds the four ways the graph with three edges emanating from a single node can be embedded in a graph with four nodes (in the above example, the complete weighted graphs corresponding to the two tetrahedra).
This is an example of an \textit{invariant graph polynomial}.
 When evaluated on a tetrahedron it encodes label-independent information of the local correlations between the four sets of three adjacent edge lengths.
 
 In a sense, $O$ is both \emph{global}, due to the group averaging $\mathcal{R}$, \emph{and local}, as it probes adjacency information about a type of connected graph substructure. We call observables of this kind \emph{glocal}. 
 In Section~\ref{sec:completeGlocal} we see that label-independent information on any finite weighted graph can be completely captured with a finite set of glocal observables.

\section{Complete Observables discern unlabelled weighted graphs}
\label{sec:unlabelledgraphs}
In this Section, we introduce the main objects of study for this work: complete sets of observables on unlabelled weighted graphs. For brevity, we sometimes call a complete set of observables a \emph{complete observable.}

\subsection{Terminology note}
To simplify the discussion, precise definitions of basic graph theoretic notions and how they relate to the formalism we use are given in Appendix \ref{graph defs}. We only deal with (simple, undirected) \emph{labeled weighted graphs} which we refer to hereafter simply as \emph{weighted graphs}, generally denoted as $G$ or $\Gamma$. The set of node (vertex) labels is fixed to be $V=\{1,...,N\}$ and $E$ is the induced set of edge labels. Any symmetric $N \times N$ matrix with real entries $G_{ij}$ both fully defines and can be seen as the adjacency matrix of a weighted graph.\footnote{Everything we do carries through also for complex weights and complex-valued invariants.} The weights are identified with the entries of the adjacency matrix such that $G_{ij}$ corresponds to the weight of the edge connecting nodes $i$ and $j$, or to the weight of the node $i$, if $i=j$. A vector space of weighted graphs is defined precisely in Section \ref{sec:InvGraphPols}.   The usual notion of a simple unweighted graph is a weighted graph with $G_{ii}=0$ for all $i$, and when $i\neq j$ we have $G_{ij}=1$ if an edge is present and $G_{ij}=0$ otherwise.  Note that any weighted graph $G$ with $N$ nodes can be thought of as a weighting of a complete unweighted graph $K_{N'}$ for $N\leq N'$ with the convention that a weight $0$ on an edge removes the edge. An action of a permutation $\sigma$ on an object $o$ is denoted as $\sigma \cdot o$.

\subsection{Passive and active picture}\label{sec2} 

A relabeling of a weighted graph $G$ is a permutation $\sigma \in \S_N$ acting on the node labels $V$ along with the induced edge label transformation ${\sigma \cdot E \coloneq \big\{\{\sigma(i),\sigma(j)\}~|~ \{i,j\}\in E\big\}}$. This assures that the adjacency relations of the graph are preserved. The induced action on the adjacency matrix is 
\begin{align}
(\sigma \cdot G)_{\sigma(i)\sigma(j)}=G_{ij}.\label{adjacency matrix trafo}
\end{align}
 That is, a change of labels corresponds to permuting the rows and columns of $G_{ij}$
 simultaneously.

A transformation that permutes the node labels can be viewed as either an active or a passive transformation. In the active picture, we imagine keeping the labels fixed in place and move the nodes along with the edges attached to them towards their new label. In the passive picture we imagine keeping the nodes and edges fixed in place, erase the old labels and write new labels, see Figure \ref{fig:act v. pass}. 

\begin{figure}[H]
    \centering
 \includegraphics[width=0.6\linewidth]{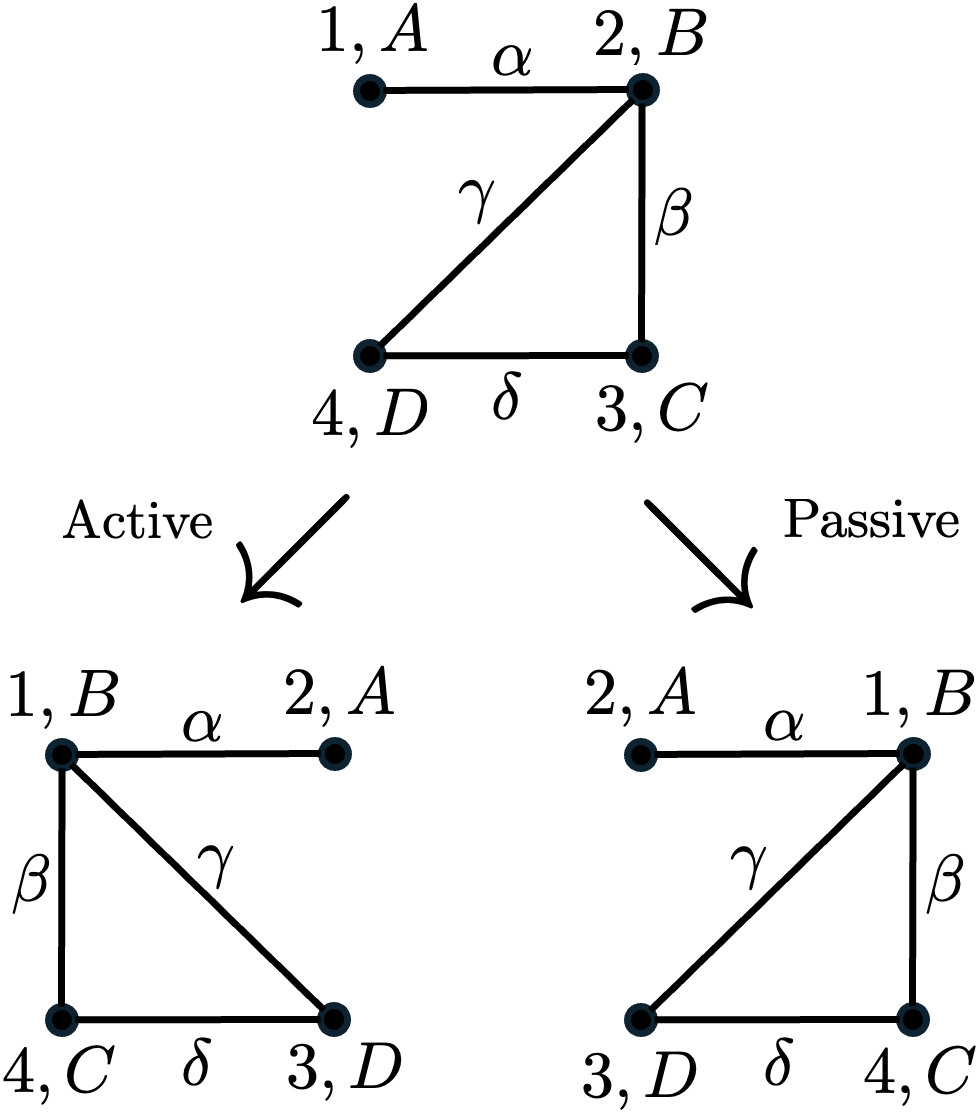}
    \caption{Transformation of a four node graph under the permutation $\sigma=(12)(34)$ in the active and passive picture. $A,B,C,D$ are node weights,  $\alpha,\beta,\gamma,\delta$ are edge weights, and $1,2,3,4$ are node labels. The active and passive picture are two equivalent ways to imagine a relabeling. The adjacency relations between the weights are preserved under a relabeling.}
    \label{fig:act v. pass}
\end{figure}

The active and passive pictures are fully equivalent. The active picture is used in the graphical calculus of Section \ref{sec:glocal}. The passive picture makes intuitively apparent that \emph{the adjacency relations between the weights are preserved under a relabeling}: while labels are changed arbitrarily, the weights `do not move'.

As an example, take the weighted graph in the top row of Figure \ref{fig:act v. pass}. It has a node weighted with $A$ and one weighted with $B$ connected by an edge weighted with $\alpha$. The relabeled graph will also have two nodes weighted with $A$ and $B$ connected by an edge weighted with $\alpha$. After a relabeling, we get an \emph{isomorphic} weighted graph. The same adjacency relations between the weights are encoded in a different node labeling convention.  

Arbitrary changes of labels do not, in general, correspond to symmetries (automorphisms) of a weighted graph. Symmetries are those permutations that leave the adjacency matrix invariant. Invariance of the adjacency matrix is not the requirement we wish to impose here. What is of interest to us is that the \emph{local relations} among the weights are left invariant, which is always true for any permutation of the node labels. This is the physical information, which is independent of the labeling. A set of complete observables fully encodes this local, relational information.

\subsection{Connection to the graph isomorphism problem}
\label{sec:graphIso}
From the point of view of physics, it may seem that the question of interest is to fix an unweighted graph and consider whether two weight assignments on it describe the same relational information. For example, the space of geometries of discrete general relativity and the kinematical space of loop quantum gravity  (discussed in Sections \ref{regge calc} and \ref{sec:spinnetworks} respectively) are typically thought to decompose into sectors on fixed unweighted graphs. 

However, a collection of functions that distinguishes any two weight assignments on any given unweighted graph, would also distinguish any two non-isomorphic unweighted graphs.\footnote{For instance, take two non-isomorphic unweighted graphs, weigh one uniformly with $2$ and the other uniformly with $3$, and take their union (or connect them by adding an unweighted edge, i.e.~weighted with $1$). This yields a weighted graph $G$. Consider also the weighted graph $G'$ with the weight assignments inverted through $2 \leftrightarrow 3$. A function that discerns $G$ and $G'$ also tells apart the two original unweighted non-isomorphic graphs. } Therefore, finding complete observables for discrete geometries defined on graphs is linked to the \emph{graph isomorphism problem} (see e.g.~\cite{thegraphisoproblem}) a notorious set of issues which, roughly, asks what the most efficient way is to decide whether two graphs are isomorphic or not.

In physics, the objects of interest are observables on a `space of weightings', states or configurations, on some arbitrary but considered fixed unweighted graph. But finding these would be equivalent to finding observables on a space of weighted graphs, where the graph connectivity can vary. Therefore, the task is to understand how complete observables are constructed for weighted graphs. In the remainder of this Section we formally define the notion of a complete observable. 

\subsection{Complete observables on weighted graphs}

Consider $\mathcal{G}_N$, the space of all weighted graphs with $N$ nodes. We define \emph{observables} to be functions 
\begin{equation}
O:~\mathcal{G}_N\longrightarrow{}\mathbb{R} \nonumber
\end{equation}
that for all $G_1,G_2\in\mathcal{G}_N$ satisfy 
\begin{equation}
    G_1\cong G_2 \implies O(G_1) = O(G_2)
    \label{def observables}
\end{equation}
A \emph{complete set of observables} or \emph{
complete observable} is a collection of observables $O_1,\dots,O_r$ that for all ${G_1,G_2\in\mathcal{G}_N}$ satisfies
\begin{align}
\label{eq:completeObservablesDefinition}
    G_1\cong G_2 \Leftrightarrow O_k(G_1)=O_k(G_2) \  \forall k=1,\dots,r
\end{align}
Here, $\cong$ signifies graph isomorphism in the sense that there exists a $\sigma\in \S_N$ such that $\sigma\cdot G_1=G_2$, they are equivalent up to a relabeling. 

We define the unlabelled weighted graph $\mathcal{O}_{G}$ as the set of weighted graphs isomorphic to $G\in\mathcal{G}_N$, that is, the \emph{orbit} of $G$ under permutations

\begin{align}
    \mathcal{O}_{G}= \{\sigma\cdot G ~|~\sigma \in \S_N\}.
\end{align}
The orbit is agnostic towards a specific choice of labels. Observables can also be seen as functions on the space of unlabelled graphs, the quotient space $\mathcal{G}_N / \S_N$. A set of complete observables discerns orbits. Taken together as one function, a complete observable is a bijection from $\mathcal{G}_N / \S_N$ to some space, which we here take to be $\mathbb{R}^r$. 

\section{Complete Glocal Observables from Invariant graph polynomials}
\label{sec:InvGraphPols}

In this Section, we construct a complete set of observables for weighted graphs as generating sets of invariant graph polynomials. We stress beforehand two facts of importance that will be shown: a complete observable (i) can always be composed of glocal observables (invariants that encode information about local correlations), and (ii) can be calculated through deterministic algorithms that terminate in finite steps.

\subsection{The vector space of weighted graphs}
\label{ref:vectorSpace}
We now introduce the vector space of weighted graphs and the algebra of polynomials in its dual basis. For $i\leq j\in\{1,..,N\}$, let $e_{[i,i]}$ be the graph with weight $1$ on node $i$ and weight $0$ everywhere else and let $e_{[i,j]}$ be the graph with weight $1$ on the edge between nodes $i$ and $j$ and weight $0$ everywhere else. Here $[\cdot]$ denotes a multiset, a set allowing for repetitions of elements. We build the vector space $\mathcal{G}_{N}$ by allowing symbolic addition and scalar multiplication so that the usual vector space axioms are satisfied. Then the set
$\mathcal{B}=\{e_{[i,j]}~|~ i\leq j \in\{1,...,N\}\}$ is a basis of $\mathcal{G}_N$. Any graph $G$ with weights $g_{[i,j]}$ is written in this basis as
\begin{align}
    G=\sum_{i\leq j\in\{1,...,N\}} g_{[i,j]}e_{[i,j]}. \label{basis exp}
\end{align}
 We fix an isomorphism between the vector space of weighted graphs and the vector space of their adjacency matrices ($N \times N$ symmetric matrices) by identifying $e_{[i,j]}$ with the matrix that has a unit in the $ij$ and $ji$ entry and all other entries zero. By abuse of terminology we call both  spaces $\mathcal{G}_N$.  With this identification the weights $g_{[i,j]}$ are the entries of the adjacency matrix $G_{ij}$ of the weighted graph $G$, i.e.
\begin{equation}
\label{eq:weightsAsAdjacencyMatrix}
g_{[i,j]}=G_{ij}=G_{ji}.
\end{equation}
Finally, note that we have $\text{dim}(\mathcal{G}_N)=\frac{1}{2}N(N+1)$.

Relabelings are implemented by defining how permutations $\sigma\in \S_N$ act on the basis elements
\begin{align}
\label{eq:SnActionDef}
    \sigma \cdot e_{[i,j]}=e_{[\sigma(i),\sigma(j)]}
\end{align}
and extending to all elements of $\mathcal{G}_N$ by linearity. This agrees with the action in Figure \ref{fig:act v. pass} and \eqref{adjacency matrix trafo}. On the vector space of weighted graphs the passive and active picture respectively can be seen as follows
\begin{equation}
\begin{aligned}
    \sigma\cdot G&=\sum_{i\leq j\in\{1,...,N\}} g_{[i,j]} \; e_{[\sigma(i),\sigma(j)]}\\
    &=\sum_{i\leq j\in\{1,...,N\}} g_{[\sigma^{-1}(i),\sigma^{-1}(j)]} \; e_{[i,j]}.\label{passive active in vector terms}
\end{aligned}
\end{equation}
 We again see the equivalence of the two pictures. We either imagine to permute the labels, corresponding to a permutation of the basis elements, or keep the labels fixed and permute the graph and weights accordingly, corresponding to the permutation acting on the weights.    

\subsection{Finite generating sets of invariant polynomials}
\label{sec:FiniteGeneratingSets}

Let $x_{[i,j]}:\mathcal{G}_N\xrightarrow[]{}\mathbb{R}$ denote the elements of the basis dual to $\mathcal{B}$. That is, $x_{[i,j]}[G]=g_{[i,j]}$, where $g_{[i,j]}$ is the coefficient corresponding to the basis element $e_{[i,j]}$ with $G$ defined in \eqref{basis exp}. Using the dual basis,  define $\mathbb{R}[\mathcal{G}_N]=\mathbb{R}[\,x_{[i,j]} \,|\,i\leq j\in \{1,...,N\}\,]$ as the algebra of all polynomials in $x_{[i,j]}$ with coefficients in $\mathbb{R}$. The permutation group $\S_N$ acts on $\mathbb{R}[\mathcal{G}_N]$ by permuting the dual basis $\sigma \cdot x_{[i,j]}=x_{[\sigma^(i),\sigma(j)]}$, which is extended to a full action on the polynomials $p\in \mathbb{R}[\mathcal{G}_N]$ by applying the above action to every appearance of a $x_{[i,j]}$ in $p$.\footnote{This definition is equivalent to $\sigma\cdot p\coloneq p\circ\sigma^{-1}$.} A polynomial $p\in \mathbb{R}[\mathcal{G}_N]$ is called invariant if $\sigma \cdot p = p$ for all $\sigma \in \S_N$. 

The subset of all invariant polynomials forms a subalgebra of $\mathbb{R}[\mathcal{G}_N]$ called the $\textit{invariant algebra}$ $\mathcal{I}^N$. It is a subset of all the observables defined by \eqref{def observables} and lies dense in the set of all permutation invariant continuous functions on graphs \cite{maron2019}. 

Polynomials in $\mathbb{R}[\mathcal{G}_N]$ can be mapped to invariant polynomials via \emph{group averaging}. The map that implements the group averaging is called the \textit{Reynolds operator}  $\mathcal{R}$, which is defined as
\begin{equation}
\begin{aligned}
    \mathcal{R}:~ \mathbb{R}[\mathcal{G}_N]&\longrightarrow \mathcal{I}^N\\
    p&\longmapsto \frac{1}{|\SN|}\sum_{\sigma \in \S_N}\sigma \cdot p.\label{eq:reynolds_def}
\end{aligned}
\end{equation}
 The Reynolds operator is linear, idempotent and satisfies the Reynolds property 
\begin{equation}
\label{eq:reynoldsProperty}
   \mathcal{R}(\mathcal{R}(p)q)= \mathcal{R}(p)\mathcal{R}(q).
\end{equation}
The Reynolds property \eqref{eq:reynoldsProperty} is key to reducing all graph invariants to glocal observables, see Section \ref{sec:reynoldsProperty}. The properties of the Reynolds operator are further detailed in Appendix \ref{sec:propertiesReynolds}.

A central result in invariant theory, Hilbert's finiteness theorem \cite{Hilbert1}, states that invariant algebras such as  $\mathcal{I}^N$ are \emph{finitely generated}. This means there exists a finite subset $I\subset\mathcal{I}^N$ called a $\textit{finite generating set}$, such that any invariant polynomial in $\mathcal{I}^N$ can be expressed as a polynomial combination of elements of $I$.

Remarkably, a finite generating set $I=\{i_1,...,i_r\}$ is also a set of complete observables \cite[Theorem 10]{inv2}. That is, for any two weighted graphs $G_1,G_2$ 
\begin{align}
    G_1\cong G_2 \Leftrightarrow i_k(G_1)=i_k(G_2) \ \forall k=1,\dots, r,
\end{align}
compare with \eqref{eq:completeObservablesDefinition}. For the convenience of the reader, the proof in our terminology is given in Appendix \ref{sec:FiniteGeneratingSetsCompleteObservables}.

A finite generating set $I$ will have polynomials up to some highest degree $d^I$. Defining $\beta(\mathcal{I}^N)$ as the smallest $d^I$ of any $I$, the following (in practice very loose) bounds are known for the action of $\SN$ on a vector space \cite{GARSIA1984107} 
\begin{align}
    \bigg\lfloor \frac{N}{2} \bigg\rfloor \leq \beta(\mathcal{I}^N) \leq \begin{pmatrix}
        D\\
        2
    \end{pmatrix},\label{degree bounds}
\end{align}
where $D$ is the dimension of the vector space, in our case  $D=\text{dim}(\mathcal{G}_N)=\frac{1}{2}N(N+1)$, so that the upper bound goes as $N^4$.

\subsection{Complete sets of glocal observables}
\label{sec:completeGlocal}
Now,  we will construct a complete set of observables (but in general not minimal, see Section \ref{sec:minimalGeneratingSets}) composed of invariant polynomials built out of \emph{connected} subgraphs or the empty graph---a complete set of glocal observables. The case of the empty graph substructure yields the observables that we have previously called purely global (see Section \ref{sec:GlobalToLocal}): it yields invariants that are averages over the node weights and encode no adjacency (edge) information. 

The techniques introduced below have been developed in the context of framing the graph isomorphism problem as the computational complexity of constructing generating sets of graph invariants. We will be following in particular \cite{thiéry2008}. The key for the construction is to use a graph to define elements of $\mathbb{R}[\mathcal{G}_N]$ and use the Reynolds operator to map them to invariant polynomials in $\mathcal{I}^N$.

Let $M\in\mathcal{G}_N$ be a graph with natural numbers as weights, i.e. 
\begin{align}
    M=\sum_{i\leq j\in\{1,...,N\}} m_{[i,j]}e_{[i,j]}, ~ m_{[i,j]}\in\mathbb{N}.
\end{align}
This is called a \emph{multigraph}. The monomial defined as
\begin{align}
    X^{M}\coloneq \prod_{i\leq j\in\{1,...,N\}} (x_{[i,j]})^{m_{[i,j]}},\label{graph monomial}
\end{align}
is called the \textit{graph monomial associated to} $M$. Graph monomials are a subset of $\mathbb{R}[\mathcal{G}_N]$ and in general they are not invariant, $X^M\notin \mathcal{I}^N$.

Applying the Reynolds operator on a graph monomial, we have that
\begin{align}
\label{eq:invariantgraphpolynomial}
    \mathcal{R}(X^M)=\frac1{|\SN|}\sum_{\sigma\in\SN} X^{\sigma\cdot M}=\frac{1}{|\mathcal{O}_M|}\sum_{H\in \mathcal{O}_M}X^H,
\end{align}
see \eqref{eq:reynolds_def}. The \emph{invariant graph polynomial} $\mathcal{R}(X^M)$ associated with $M$ is the sum over all graph monomials $X^H$ for multigraphs $H$ isomorphic to $M$, where $|\mathcal{O}_M|$ is the orbit size, given by $|\mathcal O_M|=|\SN|/|\mathrm{Aut}(M)|,$ where $\mathrm{Aut}(M)$ is the automorphism group of $M$. The invariant polynomial $\mathcal{R}(X^M)$ will be of degree $W(M),$ where $W(M)$ is the sum of the weights of $M$.

Let a \emph{quasi-connected multigraph} or \emph{local feature} be a multigraph that has either exactly one non-trivial connected component, or consists of a single weighted node (see Appendix \ref{graph defs} for full definitions). We introduce the term local feature for a quasi-connected multigraph 
$M$ both for brevity and to emphasize its role in specifying the local correlation encoded by the invariant. Recall that the trivial case when $M$ is a graph with no edges (the quasi in quasi-connected) corresponds to invariants that are purely global observables. The point of view taken here is that purely global observables are the trivial case of a glocal observable.\footnote{An intuitive analogy is $n$-point functions. The trivial case of an invariant constructed out of a multigraph with no edges is analogous to the trivial case $n=1$, single point averages that do not encode correlations.} 

We can now define \emph{glocal observables}:
\begin{definition*}[Glocal observables]\label{theorem}
    A glocal observable $O^{M}$ is a homogeneous invariant graph polynomial $\mathcal{R}(X^M)$, where $M$ is a local feature. 
\end{definition*}

\noindent With these definitions, we can state the following theorem, which characterizes a generating set made entirely of glocal observables:
\begin{theorem*}[Complete set of glocal observables]
    Let
    \begin{equation}\label{eq:theorem}
    I^N_C = \left\{ O^M ~\middle|~ W(M)\leq {\frac{1}{2}N(N+1)\choose2}\right\}.
    \end{equation}
    Then $I^N_C$ generates the full invariant algebra $\mathcal{I}^N$.
\end{theorem*}
\noindent This theorem is a generalization of Proposition 2.1 in \cite{thiéry2008}, expressed with respect to the notion of glocal observable introduced just before. Here, we have in addition to edge weights also allowed for weights on nodes, in order to allow for more generality and apply the formalism to spin networks (the volume information lives on nodes, see Section \ref{sec:spinnetworks}). The proof of the theorem in our context is given in Appendix \ref{sec:propertiesReynolds}.\footnote{Note that the bound we use is an overestimation. The highest required degree was formally defined as $\beta(\mathcal I_N)\leq{D\choose2}$ for $D = \frac12 N(N+1)$ in \eqref{degree bounds}. The number $\beta(\mathcal I^N)$ is not known a priori, it can only be learned once the algorithm that constructs a complete set of invariants is run to its termination at least once.}

 The above theorem explicitly constructs a finite generating set $I^N_C$ of $\mathcal{I}^N$, consisting of invariant graph polynomials $\mathcal{R}(X^M)$, each built exclusively out of a local feature $M$: a \emph{connected} graph structure $M$.\footnote{As remarked before, the trivial case when $M$ is a graph with no edges is also allowed and yields the purely global observables \eqref{eq:globalNodes}.} This is a complete set of glocal observables $O^M$, indexed by the set of all local features $M$ with sum of edge and node weights less than $D \choose 2$ with $D=\text{dim}(\mathcal{G}_N)$.

Glocal observables display a mix of local and global characteristics so to capture adjacency information in a way that is independent of any change of labels. The global character comes from the group averaging $\mathcal{R}$ and the local character from the local graph substructure encoded through the local feature $M$. 

\section{Glocal character of observables}
\label{sec:glocal}
In this Section we will examine through an intuitive graphical notation how glocal observables capture information on local correlations.  The graphical notation introduced below is inspired by (but does not correspond to) the graphical calculus used in \cite{thiéry2008}. We henceforth leave implicit the multiset brackets, that is, we write $x_{ij}$, $e_{ij}$ and $g_{ij}$ instead of $x_{[i,j]}$, $e_{[i,j]}$ and $g_{[i,j]}$. 
\subsection{Encoding local correlations globally}

 Consider the graph monomial
\begin{equation}
    X^M = x_{11}(x_{12})^2\in\mathbb{R}[\mathcal{G}_3]
\end{equation}
associated to the multigraph $M=e_{{11}}+2e_{12}\in \mathcal{G}_3.$
We write this as
\begin{equation}
    X^M =  \vcenter{\hbox{\includegraphics[width=0.15\linewidth]{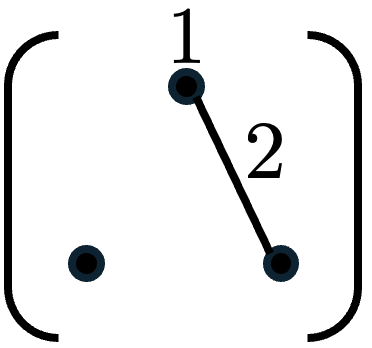}}}.\label{graphicalintro}
\end{equation}
$X^M$ is an operator on $\mathcal{G}_3$, it can be evaluated on arbitrary graphs ${G=\sum_{ij} g_{ij}e_{ij}}\in \mathcal{G}_3$. Graphically, we express this as
\begin{equation}
\begin{aligned}
    \!\!\!\!X^M[G]= \vcenter{\hbox{\includegraphics[width=0.5\linewidth]{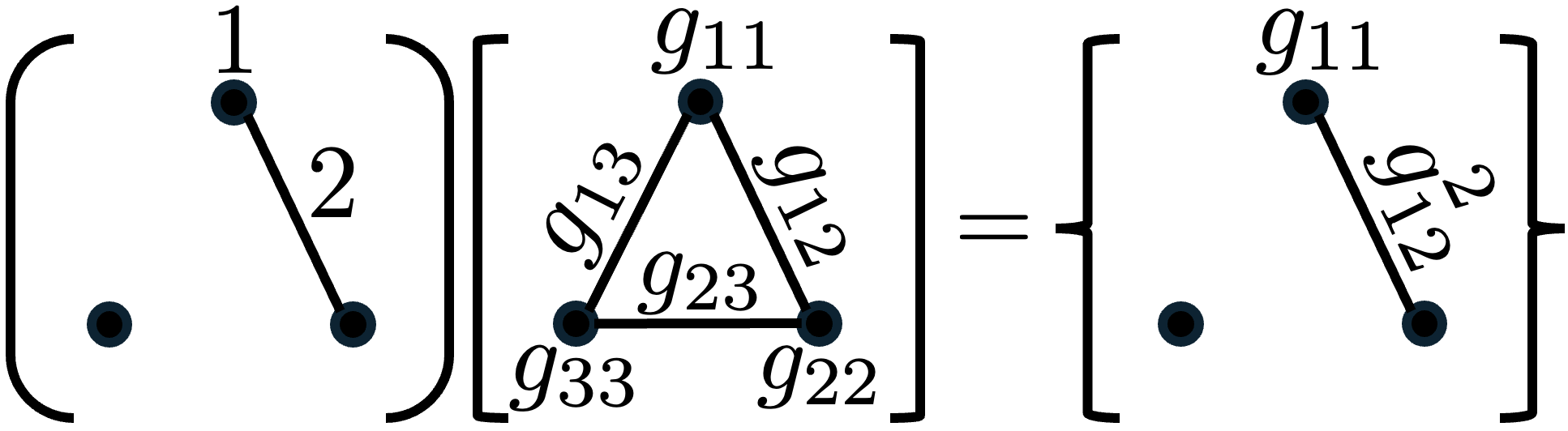}}}= g_{11}g_{12}^2.\label{ex} 
    \end{aligned}
\end{equation}
The round brackets denote the graph monomial, the square brackets indicate the graph on which it is evaluated. The curly brackets denote the multiplication of the pictured weights. 

The corresponding invariant graph polynomial $\mathcal{R}(X^M)$ in standard notation is
\begin{equation}
\begin{aligned}
    6 \,\mathcal{R}(X^M)=&\phantom{+}\, x_{11}x_{12}^2 +x_{22}x_{12}^2 
+x_{33}x_{23}^2\\
&+x_{22}x_{23}^2+x_{11}x_{13}^2+x_{33}x_{13}^2.\label{example invariant poly}
\end{aligned}
\end{equation}
Using the graphical notation given in \eqref{graphicalintro}, it can equivalently be written as
\begin{equation}    
 \label{eq:1-gloc}
 \vcenter{\hbox{\includegraphics[width=0.8\linewidth]{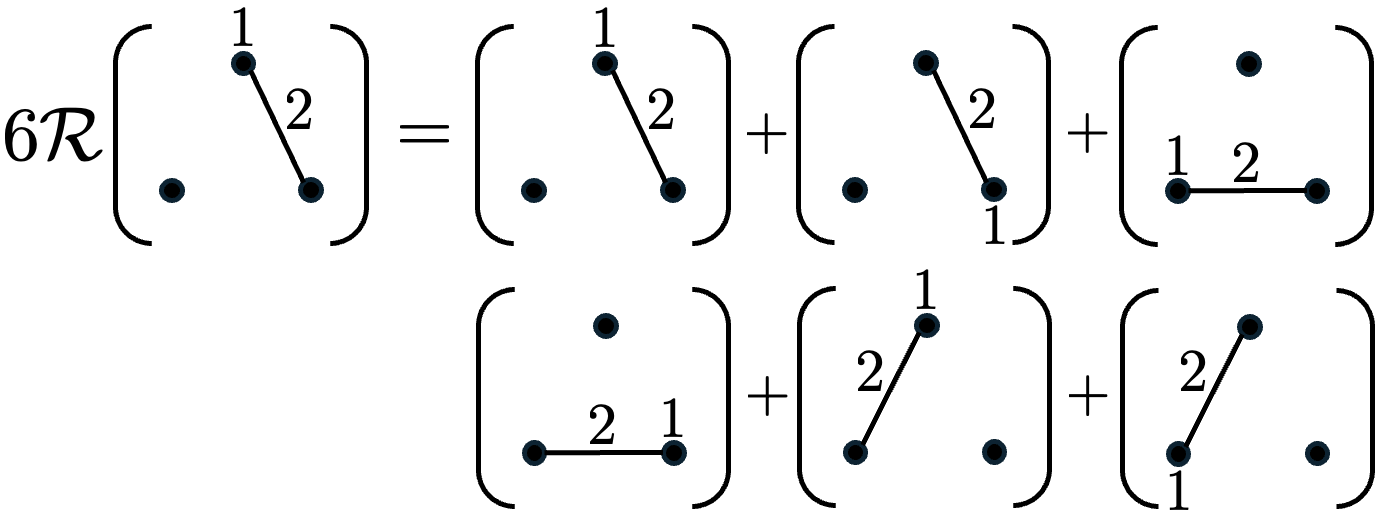}}}.
 \end{equation}
$\mathcal{R}(X^M)$ is obtained by acting with the Reynolds operator $\mathcal{R}$ on the graph monomial $X^M$ associated to the multigraph $M$, see \eqref{eq:invariantgraphpolynomial}. It is given by a sum over graph monomials $X^H$ of all multigraphs $H$ in the orbit of $M$.  
$\mathcal{R}(X^M)$ of \eqref{example invariant poly} is again an operator on the space of all weighted graphs with three nodes $\mathcal{G}_3$. An example of how it acts on graphs in $\mathcal{G}_3$ is:
\begin{equation}
\begin{aligned}
    \vcenter{\hbox{\includegraphics[width=0.7\linewidth]{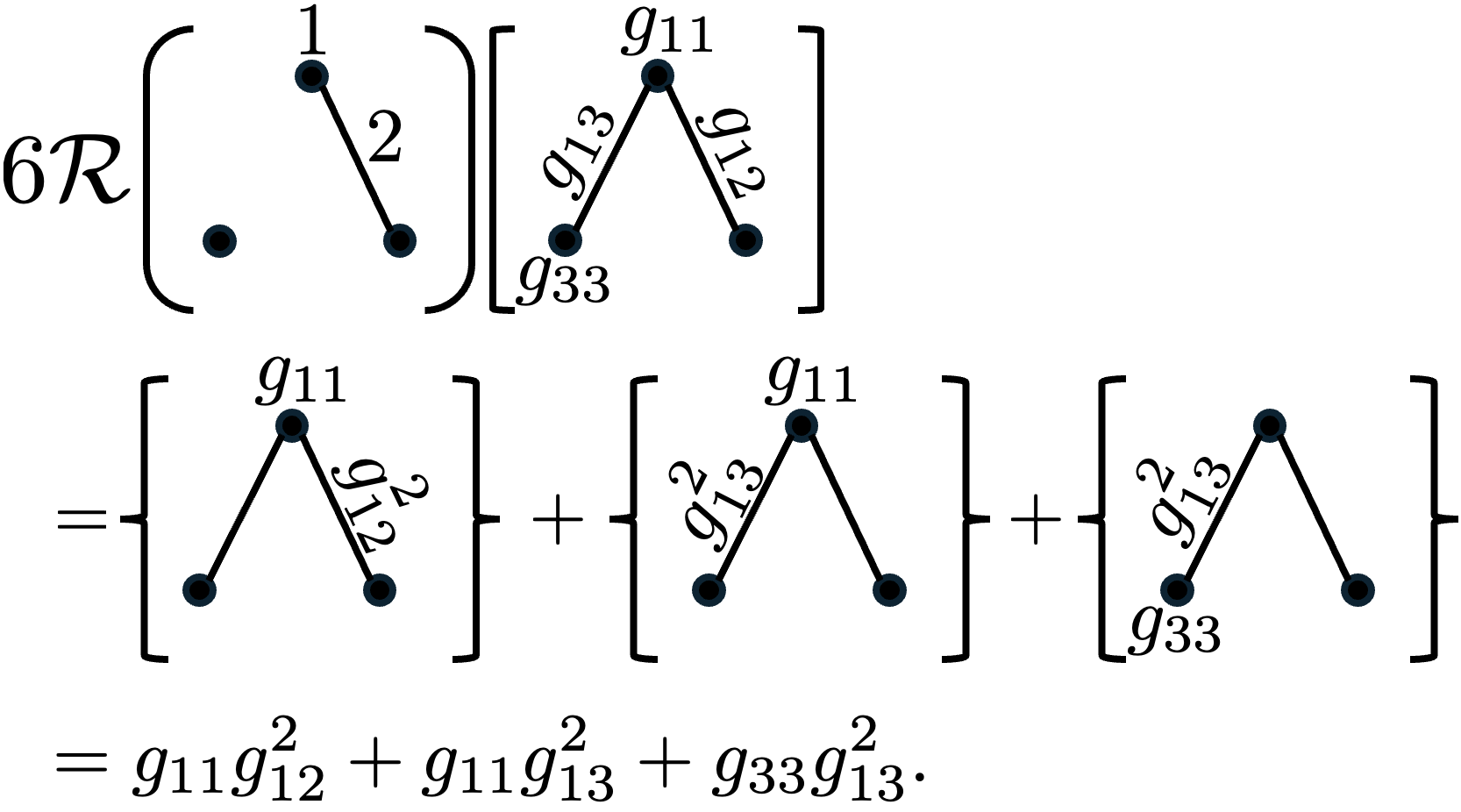}}}\label{application}
\end{aligned}
\end{equation}    
This observable is reading out the weights of ${G=g_{11}e_{11}+g_{12}e_{12}+g_{13}e_{13}+g_{33}e_{33}}$ in every possible way that the structure of $M$ (weighted node with weighted edge attached to it)  can be embedded into $G$. Note that if the adjacency structure specified by the local feature $M$ where absent on $G$, all monomials on the right side of \eqref{example invariant poly} would vanish, giving $\mathcal{R}(X^M)[G]=0$. This is general, $\mathcal{R}(X^M)[G]$ vanishes whenever $G$ does not contain at least one local correlation matching the local feature $M$.

 Any invariant graph polynomial is immediate to extend to an operator on graphs with a larger number of nodes, by adding isolated nodes with no weights in the multigraph. For example, let $H$ be the multigraph obtained from $M$ by adding one such isolated node without weight. Then $\mathcal{R}(X^H)$ will be given by
\begin{equation}
\begin{aligned}
    12\, \mathcal{R}(X^H)=&\phantom{+}\, x_{11}x_{12}^2+x_{11}x_{13}^2+x_{11}x_{14}^2\\
    & +x_{22}x_{23}^2+x_{22}x_{12}^2+x_{22}x_{24}^2\\
    &+ x_{33}x_{13}^2+x_{33}x_{23}^2+x_{33}x_{34}^2\\
&+x_{44}x_{14}^2+x_{44}x_{24}^2+x_{44}x_{34}^2.\label{extended example}
\end{aligned}
\end{equation}
All terms appearing in \eqref{example invariant poly} are still present in \eqref{extended example}. There are additional terms because we have more labels, instead of the smaller $\S_3$, the group average is over the larger $\S_4$. Graphically,
\begin{equation}
    \vcenter{\hbox{\includegraphics[width=0.8\linewidth]{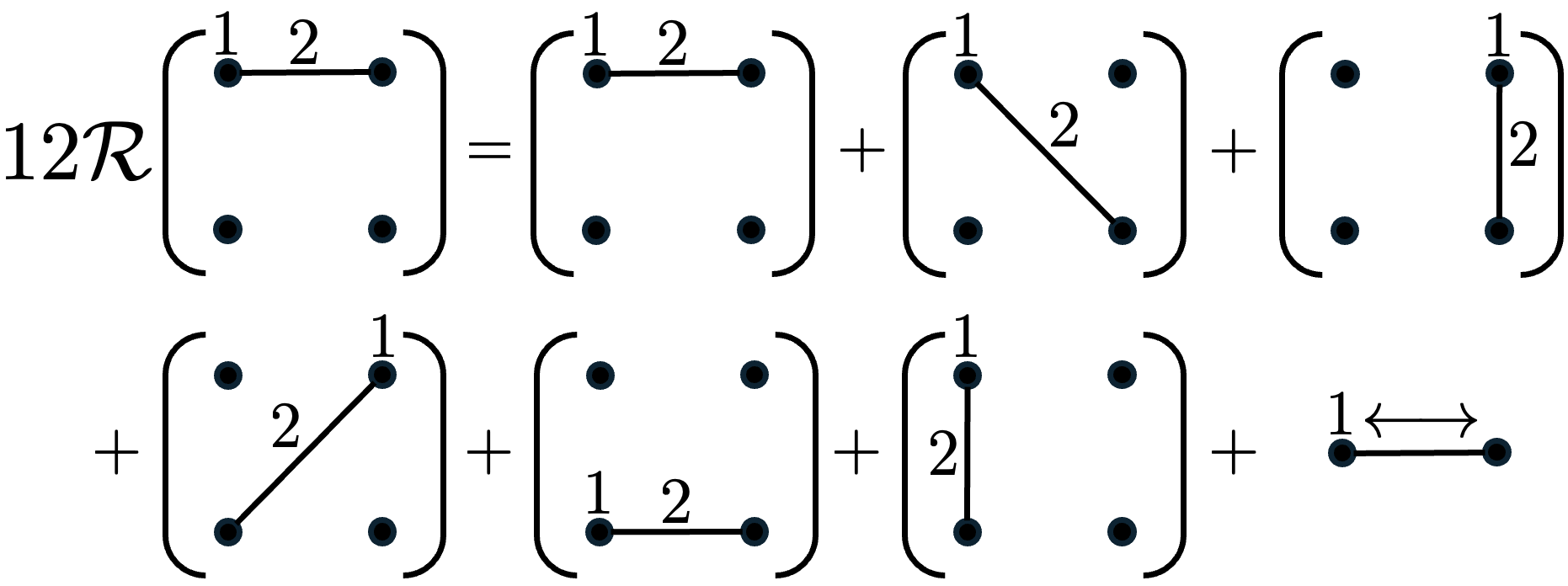}}}
\end{equation}
where we suppressed the six monomials obtained by switching the weighted node to the other end of the edge.

Let us see one more example, a glocal observable where the local feature searched is a triangle, evaluated on a weighted graph of five nodes:
\begin{equation}
\!\!\vcenter{\hbox{\includegraphics[width=0.9\linewidth]{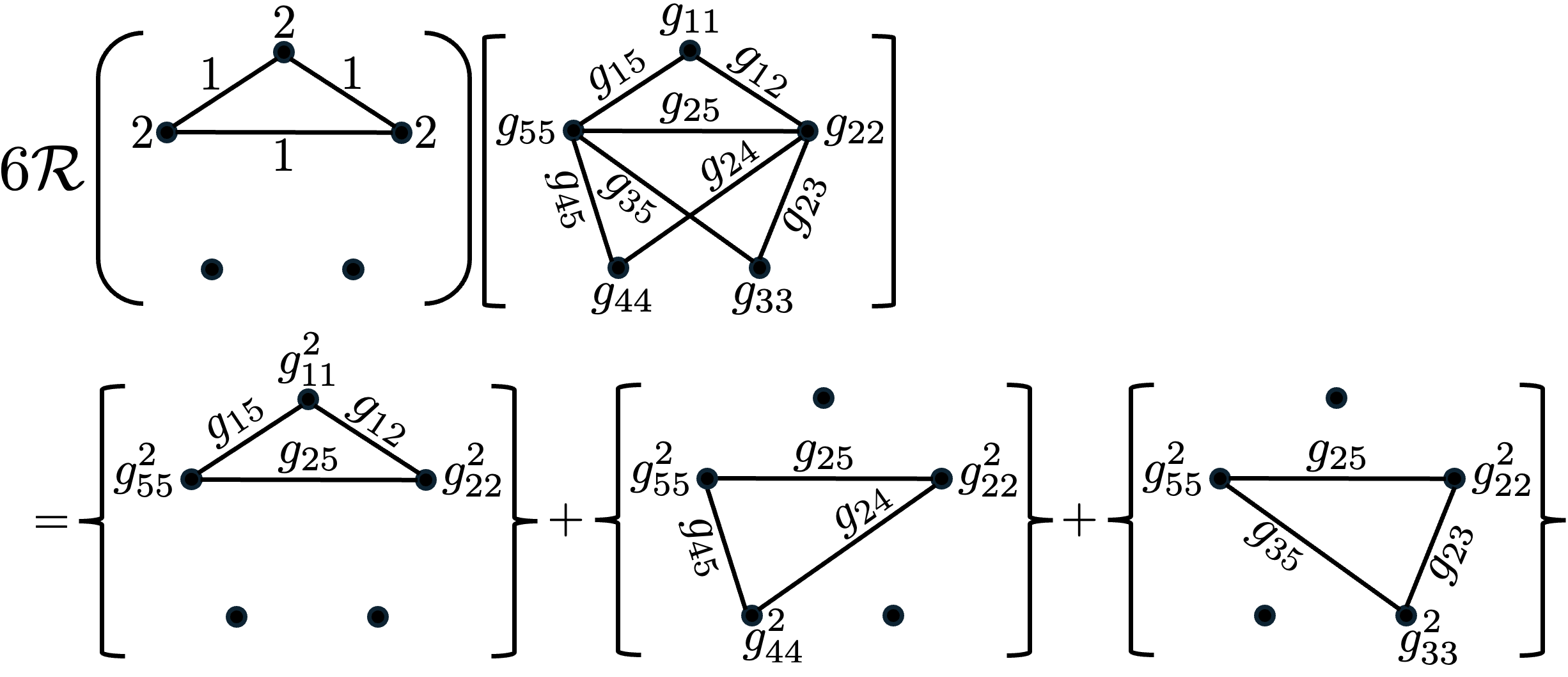}}}\label{eq:triangleObs}
\end{equation}
It returns the sum of weight correlations of the three triangle subgraphs that exist in $G$.

The above examples show that, despite their global character which arises from averaging over permutations of node labels, observables built out of connected subgraphs capture local information. Intuitively, they search everywhere on a larger graph for instances of an adjacency structure among weights, calculate each such local correlation as a product of the adjacent weights, and take the sum. This captures information about local correlations in a way that is invariant under changes of labels. For each type of local correlation, full information will be captured by constructing invariant polynomials of different degree, corresponding to multigraphs with the same underlying graph structure, but different integer weights. Repeating the procedure for multigraphs with different connectivity, the theorem in Section \ref{sec:completeGlocal} guarantees that a finite number of glocal observables captures \emph{complete} information about \emph{all} kinds of local correlations on \emph{any} finite weighted graph. In this sense, we can say that all information on a weighted graph \emph{is} glocal information.

\subsection{Why glocal correlations capture all information}
\label{sec:reynoldsProperty}
The fact that glocal observables suffice to reconstruct all information on a weighted graph is a result of how they compose under multiplication, according to the Reynolds property \eqref{eq:reynoldsProperty}. For any two multigraphs $M_1,M_2$ we have
 \begin{align}
 \label{eq:reynoldsExplicit}
    &\mathcal{R}(X^{M_1})\mathcal{R}(X^{M_2})=\frac{1}{|\mathcal{O}_{M_2}|}\sum_{H\in\mathcal{O}_{M_2}}\mathcal{R}(X^{M_1+H}),
\end{align}
where the sum runs over the multigraphs $H$ isomorphic to $M_2$. This is the sum of invariant graph polynomials obtained from all possible ways to superimpose the two graphs, divided by $|\mathcal{O}_{M_2}|$.
Note that  since $\mathcal{R}(X^{M_1}) \mathcal{R}(X^{M_2}) = \mathcal{R}(X^{M_2}) \mathcal{R}(X^{M_1})$ the roles of $M_1$ and $M_2$ on the right side of \eqref{eq:reynoldsExplicit} can be exchanged.

This remarkable property becomes intuitive in the graphical calculus. Take the following example\footnote{We have suppressed coefficients that appear in the Reynolds product, see Appendix \ref{sec:propertiesReynolds}. They can be written in  a closed but impractical form; see \cite[Section 2]{mikkonen2008algebragraphinvariants}.}  
\begin{equation}
\vcenter{\hbox{
\includegraphics[width=0.7\columnwidth]{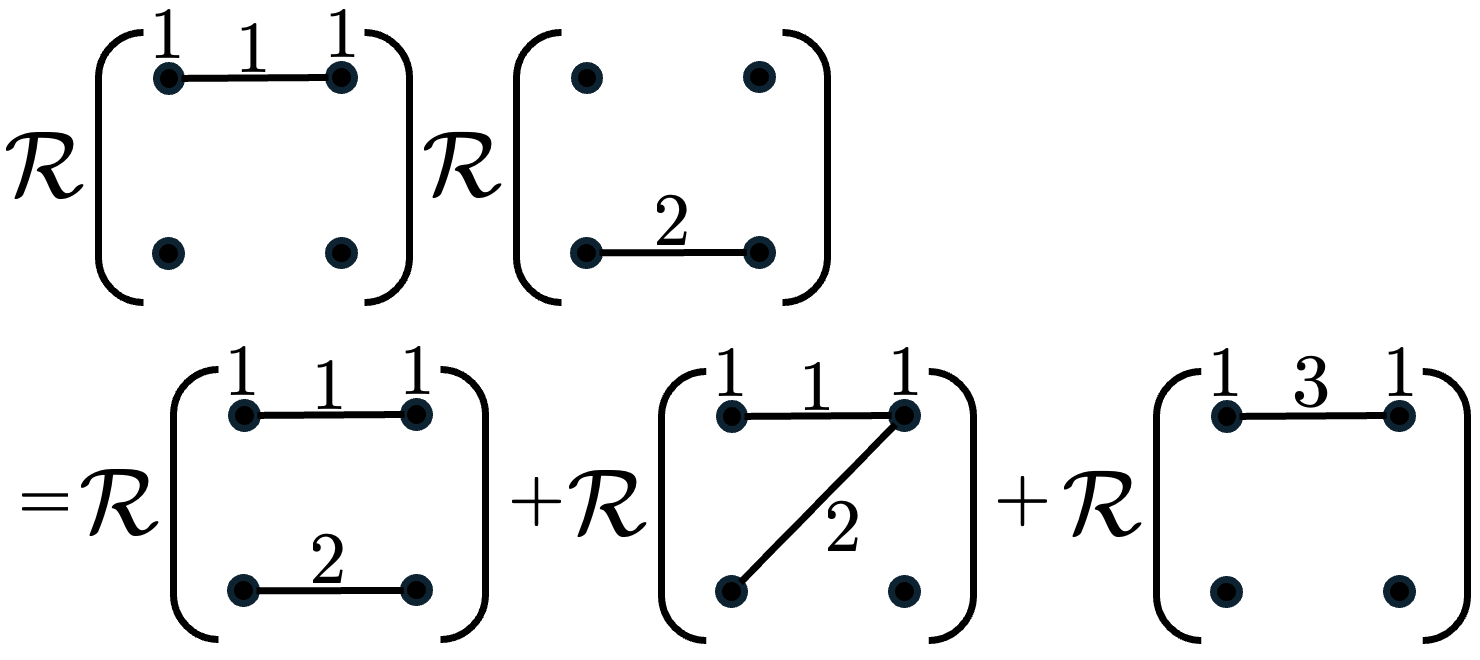}}}.\label{superimposition example}
\end{equation}
On the right hand side, we have one term for every possible non--isomorphic superimposition of the two multigraphs. This behaviour is natural. The product of two invariant graph polynomials must yield again an observable, an invariant polynomial living in $\mathcal{I}^N$. Observables are functions on unlabeled weighted graphs---the orbits, they are agnostic about the labeling. When multiplying two invariant graph polynomials, since the labels have been erased, no notion of `alignment' of the two graphs exists. Then, the new invariant takes into account all relative graph alignments.

Now, note that the first $\mathcal{R}(X^M)$ on the right side of \eqref{superimposition example} is built out of a \emph{non--connected} multigraph, it has two connected components. A moment of reflection shows that when the left side has two connected multigraphs, the right side will have only terms built on connected multigraphs, and \emph{at most one} multigraph with two connected components. This is general: up to isomorphism, there is always \emph{at most one} superimposition of two multigraphs which has the sum of the number of connected components of each. Every other superimposition will have at least one less connected component. Therefore, the Reynolds property allows to construct relations to trade any non glocal observable appearing in a generating set with glocal observables.

\subsection{Most information is not purely global}
\label{sec:minimalGeneratingSets}
 The finite generating set described by the theorem given in Section \ref{theorem} in general will not be minimal. Observables in a minimal generating set cannot be reduced to one another through algebraic operations. 
 
 There exist algorithms that guarantee explicit computation of minimal generating sets in finite steps. The most efficient known at the moment is King's algorithm \cite{king2007}. It is a remarkably simple algorithm for directly computing minimal generating sets of algebras of invariants. Using an implementation of King's algorithm in the computer algebra system SageMath given in \cite{masters}, we explicitly computed a minimal generating set for $\mathcal{I}^3$ in seconds and for $\mathcal{I}^4$ in a few minutes, including weights on nodes. These are given in Appendix \ref{appendix: minimal generating set}. The minimal generating set for $\mathcal{I}^4$ consists of 31 homogeneous invariant polynomials, with the highest degree appearing being 5.\footnote{Apart from the bounds given in \eqref{degree bounds}, it is not known a priori (before its explicit calculation is carried out) how many elements a minimal generating set will have, and what their exact smallest degree bound will be. However, all minimal generating sets will have the same number of elements, and, in particular, the same number of elements of a given degree. This also implies that in any minimal generating set, the highest degree polynomials will always be of degree equal to the smallest degree bound.}  The algorithm appears to scale very fast with the number of nodes $N$, computation of a minimal generating set for $\mathcal{I}^5$ seems to require several weeks for an ordinary computer.
 
The minimal generating set for $\mathcal{I}^3$ includes the six purely global observables capturing the information of the three edge weights and the three node weights, corresponding to  \eqref{eq:powersumNodes} and \eqref{eq:powersumEdges} for $K=1,2,3$. The three other glocal observables in the minimal generating set for $\mathcal{I}^3$ capture the local correlations corresponding to a node weight adjacent to an edge weight. One of them is:
\begin{equation}
    \!\!\vcenter{\hbox{\includegraphics[width=0.7\linewidth]{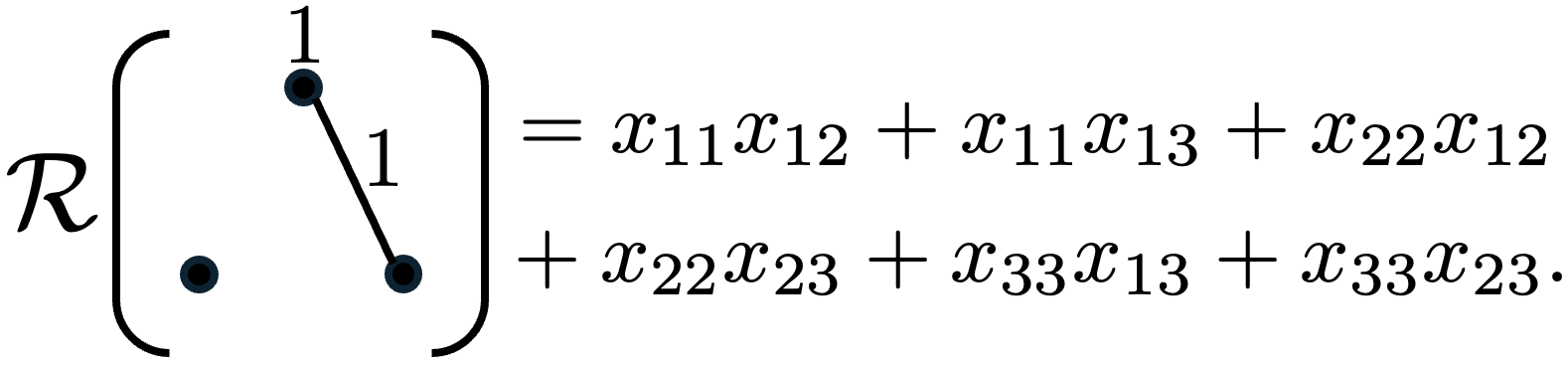}}}
\end{equation}
This glocal observable can not be constructed from purely global observables of the form \eqref{eq:powersumNodes} and \eqref{eq:powersumEdges}. This can be seen explicitly by inspection of the minimal generating set: since this is a degree two polynomial, the only possibility to construct it from purely global quantities is through combinations of the degree one $O^E_1$ and $O^V_1$, which are in the same minimal generating set. In general, because the space of weighted graphs with $N$ nodes is a subspace of that with $N'>N$ nodes (see Section \ref{ref:vectorSpace}), complete observables consisting of only purely global observables do not exist for any $N>1$.

As another example of this behavior, the fact that purely global observables cannot generate the non-trivial glocal observables can be seen in the minimal generating set of $\mathcal{I}^4$ that we have calculated. The glocal observable $P_4 \propto \sum_\sigma x_{\sigma(1)\sigma(2)} x_{\sigma(1)\sigma(3)}$, which captures information on correlations between adjacent pairs of edge weights, according to polynomial degree could only have been generated by the global $Q_2 \propto \sum_\sigma x_{\sigma(1)\sigma(2)} $. Since they both appear in the same \emph{minimal} generating set, this is impossible. 

It is natural to expect that on a large weighted graph, \emph{most} invariants in a complete set of glocal observables will not be purely global. A quick way to be convinced of this, is to recall that on unweighted graphs the purely global observables only give the number of edges and the number of nodes. On a weighted graph, the number of independent purely global observables scales with $N^2$, which corresponds to the number of edge and node weights in the graph. This is a sufficient number of parameters to reconstruct the \emph{unordered} set of weights on the graph, not their adjacency relations. The connectivity of the graph remains to be captured, and this must be done with the non-trivial glocal observables.

The cardinality of a minimal generating set is expected to grow fast with $N$. How fast, is a difficult combinatorial question.\footnote{Explicit lower bounds for minimal generating sets of invariant graph polynomials do not appear to be known. Exponential lower bounds for the cardinality of minimal generating sets have been shown in other areas of invariant theory \cite{kadish2011countinggeneratinginvariantssemisimple, derksen2019exponentiallowerbounddegrees}.} A rough argument for what happens is as follows. The number of possible local kinds of correlations is the number of connected unlabeled subgraphs up to $N$ nodes. This scales asymptotically as $2^{N \choose 2}/N!$, which is super-exponential in $N$. The difficulty is to determine how many of these local features are expected to yield algebraically independent invariants, since, glocal observables can be reduced to others built on smaller connected graphs via the Reynolds relation \eqref{eq:reynoldsProperty}. However, a Reynolds property can only be used for this if it yields a sum of only glocal observables, that is, invariants built out of (quasi-)connected subgraphs (that we have termed local features $M$). For two graphs with  number of nodes in the connected components $N_1$ and $N_2$ such that $N_1 + N_2 \leq N$, the Reynolds product will yield at least one term corresponding to a \emph{disconnected} graph. Therefore, for large $N$, such relations cannot in general be algebraically combined to remove the disconnected graph invariants \cite{mikkonen2008algebragraphinvariants}.

The above imply that for large $N$ a minimal generating set will include at least one glocal observable built for \emph{each} connected subgraph of the complete graph with $N/2$ nodes. Then, the number of glocal observables needed in a minimal generating set is expected to grow super exponentially with $N$.\footnote{This could be an overestimation, also in view of the fact the evaluation of a complete observable on two graphs solves the graph isomorphism problem, for which a pseudo-polynomial algorithm has been found \cite{babai2016graphisomorphismquasipolynomialtime}.  } 

What seems clear is that there will be many more non-trivial glocal observables than purely global observables in a minimal generating set: an arbitrary large weighted graph can encode a vast amount of local information.

\section{Discrete Spacetime}\label{regge calc}
We now apply the techniques developed in the previous sections to discrete classical and quantum spacetime. We will see that: (i) glocal observables directly provide a complete set of observables for Regge calculus, (ii) promoting glocal observables to operators yields complete sets of commuting loop quantum gravity operators.
\subsection{Discrete General Relativity}
\label{sec:Regge}
Regge calculus \cite{Regge:1961px} is a discretization of general relativity which is expected to approximate a spacetime geometry to arbitrary precision \cite{Brewin:2000zh,Gentle:2002ux}.   In this approach, instead of having a manifold with a metric,  the theory is defined on a simplicial complex $\mathcal{C}$ whose skeleton, its edges, are weighted with edge lengths $l_e$. The theory is defined through the Regge action \cite{Dittrich:2021gww}
\begin{equation}
\begin{aligned}
\label{eq:ReggeAction}
   S^\text{R}_\mathcal{C}[l_{e \subset \mathcal{C}}] =& \sum_{t \subset \mathcal{C \setminus \partial C}} A_t(l_{e \subset t}) \epsilon_t(l_{e \subset t}) \\
 &+ \sum_{t \subset \partial \mathcal{C}} A_t(l_{e \subset t}) \psi_t(l_{e \subset t}) \\
       &-  \Lambda \sum_{s \subset \mathcal{C}} V_s(l_{e \subset s}). 
\end{aligned}
\end{equation}
 The first term is the bulk term, the second is the boundary term, and we have included the cosmological term with cosmological constant $\Lambda$. The edge lengths $\l_e$ correspond to a discrete metric. $A_t$ is the area of the triangle $t$ and $V_s$ is the 4-volume of a 4-simplex  $s$. Triangles on the boundary are signified as $t \subset \partial \mathcal{C}$ and triangles in the bulk as $t \subset \mathcal{C \setminus \partial C}$. The discrete curvature is encoded in deficit angles $\epsilon_t$ and $\psi_t$. All non constant quantities in the action are fully determined by the edge lenghts $l_e$.

The details of how the quantities of this action are constructed are not of essence here, we only need to observe that the action \eqref{eq:ReggeAction} is manifestly invariant under changes of labels. This must be the case as the action is a functional of discrete geometries, and geometry does not depend on the choice of labels.  

The equations of motion of Regge calculus are obtained by imposing stationarity of the action under variations of the edge lengths $l_e$ assigned to the skeleton of the simplex $\mathcal{C \setminus \partial C}$ at fixed boundary data---\emph{a weighted graph} $G_\mathcal{C}$ with edge weights the lengths $g_e=l_e$. Since the Regge action is invariant under changes of labels, \emph{it is a functional on unlabeled weighted graphs} $\mathcal{O}_{G_\mathcal{C}}$. These are the orbits of weighted graphs $G_\mathcal{C}$ which are skeletons of piecewise flat simplicial complexes $\mathcal{C}$. Skeletons $\mathcal{C}$ with $N$ nodes form a subset  $\mathcal{G}^\mathcal{C} \subset \mathcal{G}_N$ of the space of weighted graphs with $N$ nodes. Then, every solution of the equations of motion corresponds to a weighted graph $G^\text{R}_\mathcal{C}$. The space of solutions of Regge calculus is a subset $\mathcal{G}^{\mathcal{C},\text{R}} \subset \mathcal{G}^{\mathcal{C}} \subset \mathcal{G}_N$. Then, for any $N$, a \emph{Regge geometry} is an orbit $\mathcal{O}_{G^\text{R}_\mathcal{C}}$ with $G^\text{R}_\mathcal{C} \in \mathcal{G}^{\mathcal{C},\text{R}}$.

It follows that all the information determining a Regge geometry on a simplicial complex with $N$ nodes is encoded in a complete set of glocal observables $I^N_C$ for $\mathcal{G}_N$ per the theorem of Section \ref{sec:completeGlocal}. This is because a complete set of observables on $\mathcal{G}_N$  is also complete on any subset of $\mathcal{G}_N$. In other words, since a complete set of observables discerns all the orbits in $\mathcal{G}_N$, it will also distinguish Regge geometries $\mathcal{O}_{G^\text{R}_\mathcal{C}}$ which are some of these orbits.  

Therefore, finite generating sets of the algebra of invariant graph polynomials provide a solution to the problem of observables for discrete general relativity defined on finite simplicial complexes. In particular, the theorem of Section \ref{sec:completeGlocal} defines finite sets of glocal observables that \emph{completely} describe any finite Regge geometry.

\subsection{Quantum geometry}
\label{sec:spinnetworks}
Let us consider a truncation of the kinematical Hilbert space of loop quantum gravity (LQG) on graphs for up to $N$ nodes. This Hilbert space can be decomposed (or is defined) as a direct sum over graphs $\gamma$,
\begin{equation}
\mathcal{H}_N^\mathrm{LQG} = \bigoplus_{\gamma}^ {|V_\gamma| \leq N} \mathcal{H}_\gamma
\label{eq:sum-gamma}
\end{equation}
where $\gamma$ runs over all unweighted graphs with up to $N$ nodes; see for instance \cite{giesel_2017,sahlmann_2010,ashtekar_lewandowski_2004,perez_2003,thiemann_2003, rovelli_vidotto_2015, ashtekar_bianchi_2021}.  Because every unweighted graph can be seen as some weighting of the complete graph, there is no need to consider a direct sum decomposition into Hilbert spaces built on individual unweighted graphs. As we see in a moment,  \emph{for any given $N$ it suffices to only consider a single Hilbert space built on the complete graph}.

The spin network states $\ket{\gamma;j_l,v_n}$ give a basis for each sector $\mathcal{H}_\gamma$. These are the states that simultaneously diagonalise the areas $\hat A_{l}$ and the oriented volumes $\hat V_n$, for each link (edge) $l$ and node (vertex) of $\gamma$. The $j_l\in\NN/2$ correspond to representations of $\mathrm{SU}(2)$ and give the area $A_l = 8\pi\gamma\ell_\mathrm{P}^2\sqrt{j_l(j_l+1)}$. The $v_n$ are eigenvalues of the oriented volume operator, indexing a basis of the intertwiner space of the $\mathrm{SU}(2)$ representations flowing into the node. The spin network states $\ket{\gamma;j_l,v_n}$ describe quantum geometries roughly in the sense that while their areas and volume are sharp, other geometrical quantities (for instance, dihedral angles) have a finite quantum uncertainty.

Now, recall that a weighted graph can be thought of as a weighting of the complete graph over $N$ nodes, where a weight zero on an edge is understood as removing the edge, see \eqref{basis exp}. More precisely, any weighted graph with up to $N$ nodes is an element of the vector space of weighted graphs $\mathcal{G}_N$, a basis of which are the graphs which correspond to each of the nodes and each of the edges of the complete graph. This construction fits nicely into the structure of spin networks, with the zero in the spectrum of the area operator understood as removing an edge. 

Then, instead of using a direct sum over graphs as in \eqref{eq:sum-gamma},  the Hilbert space $\H^{\mathrm{LQG}}_N$ can be formulated directly in terms of the complete graph with $N$ nodes. In fact, a spin-network $\ket{\gamma;j_l,v_n}$ on a graph $\gamma$ is orthogonal to a spin-network on a larger graph $\gamma'$, unless the difference is the presence of nodes and edges with $0$ weight, which corresponds to no added areas or volumes. We can instead consider spin networks on the complete graph $K_N$
\begin{equation}
    \H^\mathrm{kin}_N = \H_{\mathrm K_N}.
\end{equation}
 The weighted graphs corresponding to spin networks are a subset $\mathcal{G}^\text{qg}_N \subset \mathcal{G}_N$ of the space of weighted graphs. Then, a spin network state in $\H^\mathrm{kin}_N$ corresponds to a graph with adjacency matrix
\begin{equation}
\Gamma_{nn}=v_{n},~~~~~~~~~~ \Gamma_{nm}=A_{l(nm)},
\end{equation}
where $l(nm)$ is the link\footnote{Traditionally, one assigns a spin label $j_l$ to each edge, but the relation between $A_l$ and $j_l$ is invertible, and using the areas directly will simplify the construction.} between nodes $n$ and $m$. Thus, we denote a spin network state simply as $\ket{\Gamma}$. The standard inner product on $\mathcal{H}_N^{\text{kin}}$ is given by
\begin{equation}
\label{eq:innerProduct}
 \braket{\Gamma | \tilde{\Gamma}} =
\begin{cases}
1, & \text{if } \Gamma = \tilde{\Gamma}  \\
0, & \text{if } \Gamma \neq \tilde{\Gamma}
\end{cases}
\end{equation}
for any  $\Gamma,\tilde{\Gamma} \in \mathcal{G}^\text{qg}_N$.

The graph monomials $X^M$ can be naturally promoted to self-adjoint operators $\hat X^M$ on $\H^\kin_N$ by defining their action on the spin-network basis as
\begin{equation}
    \hat X^M\ket{\Gamma}=X^M[\Gamma]\ket{\Gamma},
\end{equation}
which, by linearity, extends to a generic state ${\ket\psi=\sum_{\Gamma}\psi(\Gamma)\ket\Gamma,}$
as
\begin{equation}
    \hat X^M \ket{\psi} = \sum_{\Gamma\in\mathcal{G}_N^{\text{qg}}} \psi(\Gamma)X^M[\Gamma]\ket\Gamma.
\end{equation}
Every graph monomial operator $\hat X^M$ can be written in terms of the area $\hat A_{l}$ and volume $\hat V_i$ operators as
\begin{equation}
    \hat X^M = \prod_{i\leq j=1}^N (\hat V_{i})^{m_{[i,i]}}\cdot (\hat A_{l(ij)})^{m_{[i,j]}},
\end{equation}
where $m_{[i,j]}\in\NN$ are the weights of the multigraph $M$. There are no operator-ordering ambiguities, because all the area and volume operators commute. By definition, all the graph-monomial operators are simultaneously diagonalisable.

Now, the spin--network states $\ket{\gamma;A_l,v_n}=\ket{\Gamma}$ are clearly not invariant under changes of node labels, as they are defined via a \emph{labeled} weighted graph. The permutation group $\SN$ acts unitarily\footnote{Note this definition follows from the fact that the action~\eqref{eq:SnActionDef} on the labeled graphs is a right action.} on $\H^\kin_N$ as
\begin{equation}
    U_\sigma\ket{\Gamma} = \ket{\sigma^{-1}\cdot\Gamma}.
\end{equation}
We have argued that the invariance under permutations of node labels is the natural implementation of background independence in a discrete setting. Therefore, the relabeling invariant subspace
\begin{equation}
\H_N^{\text{sym}}=\Inv_{\SN}\H_N^\kin
\end{equation}
is a more natural candidate on which to build the kinematical sector of loop quantum gravity.\footnote{Some of the authors of the present work are studying the permutation invariant quantum polyhedron~\cite{polyhedra,dibiagio2025permutation}. Permutation invariance has nontrivial effects, reducing the number of available geometries and modifying the spectrum of the volume operator---the quantum geometry exclusion principle---while maintaining a well--behaved semiclassical limit.}

A physical observable should commute with the action of $\SN$. Note that the area and volume operators of a link or a node are not $\SN$ invariant, as they are defined in terms of the labeling. For example,
\begin{equation}
    U^\dagger_\sigma \hat V_i U_\sigma=\hat V_{\sigma(i)}.
\end{equation}
However, it is immediate to promote functions of these operators to self-adjoint operators on $\H^\kin_N$ via the invariant graph polynomials $O^M$, through the definition
\begin{equation}
    \hat O^M\ket{\Gamma}=O^M[\Gamma]\ket{\Gamma}.
\end{equation}
These self--adjoint operators commute with the action of the permutation group,
\begin{equation}
    [\hat O^M,U_\sigma]=0,
\end{equation}
which can be seen either by invoking the relabeling invariance of the graph polynomials,
\begin{equation}
\begin{aligned}
\hat O^M U_\sigma \ket{\Gamma}
&=O^M[\sigma^{-1}\cdot \Gamma] \ket{\sigma^{-1}\cdot \Gamma} \\
&=O^M[\Gamma] \ket{\sigma^{-1}\cdot \Gamma}\\
&=O^M[\Gamma]U_\sigma \ket{ \Gamma}\\
&=U_\sigma \hat O^M[\Gamma]\ket{\Gamma},
\end{aligned}
\end{equation}
or by noting that
\begin{equation}
    \hat O^M = \frac1{N!}\sum_{\sigma\in\SN} \hat X^{\sigma \cdot M} = \frac{1}{N!}\sum_{\sigma\in\SN}U_\sigma^\dagger \hat X^M U_\sigma.
\end{equation}

An orthonormal basis for the relabeling--invariant subspace $\H^\mathrm{sym}_N$ is defined as follows. The projector
\begin{equation}
    P^\text{sym} = \frac{1}{N!}\sum_{\sigma \in \SN} U_\sigma. 
\end{equation}
sends each spin network state $\ket{\Gamma} \in \H_{N}^\text{kin}$ to the relabeling invariant state
\begin{equation}
    P^\text{sym} \ket{\Gamma}  = \frac{1}{N!} \sum_{\sigma \in \SN} \ket{\sigma \cdot \Gamma}
     = \frac{1}{|\mathcal{O}_{\Gamma}|}\sum_{H\in \mathcal{O}_\Gamma} \ket{H}.
\end{equation}
This is the uniform superposition of all the spin networks states corresponding to graphs in the orbit of $\Gamma$. Defining the normalised states
\begin{equation}
     \ket{\mathcal{O}_\Gamma} \coloneq \frac{P^\text{sym} \ket{\Gamma}}{\matrixel{\Gamma}{(P^\text{sym})^\dagger \,  P^\text{sym}}{\Gamma}}  \in \H_N^\mathrm{sym},
\end{equation}
we have
\begin{equation}
    \ket{\mathcal{O}_\Gamma} = \frac1{\sqrt{|\mathcal O_\Gamma|}}\sum_{H\in\mathcal O_\Gamma}\ket{H},
\end{equation}
with $|\mathcal O_\Gamma|=|\SN|/|\mathrm{Aut}(\Gamma)|.$
This basis of states satisfies the following orthonormal relations
\begin{equation}
\label{eq:innerProductSym}
 \braket{ \mathcal{O}_{\Gamma} | \mathcal{O}_{\tilde{\Gamma}} } =
\begin{cases}
1, & \text{if } \Gamma \cong \tilde{\Gamma}  \\
0, & \text{if } \Gamma \ncong \tilde{\Gamma}
\end{cases}
\end{equation}
where now the cases are decided based on \emph{isomorphism} between weighted graphs. Compare to \eqref{eq:innerProduct} where the cases are decided based on equality as adjacency matrices.

We can call the states $\ket{\mathcal O_\Gamma}$ \textit{unlabelled spin network states.} They form an orthonormal basis for $\H^\mathrm{sym}_N$. Therefore, a generic observable can be defined in terms of its matrix elements as
\begin{equation}
\label{obsQ}
\hat O = \sum_{\mathcal O_\Gamma,\mathcal O_{\Gamma'}}f(O_\Gamma,\mathcal O_{\Gamma'})\ket{\mathcal{O}_\Gamma}\!\!\bra{\mathcal{O}_{\Gamma'}}
\end{equation} with $f(O_\Gamma,\mathcal O_{\Gamma'})=f(O_\Gamma',\mathcal O_{\Gamma})^*$.
When $\Gamma \ncong \Gamma' $, these observables are graph changing. Therefore, they allow to study transitions between unlabelled spin networks on graphs with different connectivity. 

Finally, per the theorem in Section \ref{theorem}, the set of invariants $\hat O^M$, with $W(M)$ high enough will form a complete set of commuting observables for $\H^\mathrm{sym}_N$. Their simultaneous eigenvalues uniquely identify a state $\ket{\mathcal O_\Gamma}$ in the orthonormal basis of unlabelled spin network states. Further, it is not hard to show that the set of quantum observables of the form \eqref{obsQ} provides a tomographically complete set of observables, because all the invariant operators that do not commute with $\hat O^M$ are of that form, for $\Gamma \ncong \Gamma' $. This fact will be detailed elsewhere \cite{EmilMitrakos}. 

Given that geometry cannot depend on changes of labels, and the same should apply for quantum geometry, the structure described above strongly suggests a reformulation of the kinematics of loop quantum gravity based on $\H^\text{sym}_N$. In particular, working with a Hilbert space that is the span of unlabelled spin networks is a discrete and quantum analogue of imposing the spatial diffeomorphism constraint of general relativity.

\section{Comments on the continuum}
We have seen that glocal observables provide a constructible and physically meaningful solution to the problem of observables for weighted graphs with a finite number of nodes,  demonstrating how local information is fully encoded in appropriately chosen sets of global invariants. Now, we give a heuristic discussion on the possible relevance of our findings to continuous geometries.

\subsection{Discretised diffeomorphisms?}\label{all conventions vs. some conventions}

At their most basic, coordinates are an assignment of labels to points of the mathematical substrate on which we define a physical theory. Starting from this bare bone structure, one then introduces additional mathematical and physical structures. This is the philosophy of a background independent theory, where no notion of spacetime exists beforehand. 

\emph{Not all} changes of labels will be allowed on a graph if we think of them as anything more than arbitrary enumerations of mathematical points. At first thought, restricting the allowed label conventions might seem to make things \emph{easier} since the symmetry group would be \emph{smaller}. But, the action of our `group of allowed changes of labels' will be now more \emph{complicated} due to the additional rules it needs to satisfy. Then, this is obscuring a \emph{simpler} picture.

Discrete versions of the diffeomorphism action will be ambiguous \cite{bahr_gambini_pullin_2012, Bahr_2009,bahr_dittrich_2009, freidel_louapre_2003}, they approximate inherently continuum notions on a discrete structure. If we do restrict to some `allowed labelings' that are in some sense a discrete version of diffeomorphisms, these will necessarily be \emph{some} of the permutations. Therefore, the invariants we have studied here would \emph{also} be observables under \emph{any} discretised version of diffeomorphisms.

What are we then to make of `extra' objects invariant under some subset of changes of labels, but not invariant under  all changes of labels? These are things that will depend on whether we have called a node of a graph `42' or `6', how can these be `true things'? The meaningful observables of a discrete theory are those that are invariant under all changes of labels. 

In fact, it is natural to \emph{define} the geometrical information on an arbitrary weighted graph as the information held by the invariants that fully specify it up to isomorphism. In the continuum, we must restrict the allowed choices of labels in order to preserve differentiability, so we can use differential geometry. In the discrete, we are in a position to push the central insight of general relativity to the fullest, and allow for all labeling conventions.

\subsection{Complete graphs replace manifolds} 

In general relativity, there is not one group of diffeomorphisms but a group of diffeomorphisms $\phi_\mathcal{M}$ for each differentiable manifold $\mathcal{M}$, the autormorphisms which preserve its differential structure. Here, we have always considered the action of $\SN$ on weighted graphs, but a generic weighted graph \emph{does not} have $\SN$ as its automorphism group. $\SN$ is the automorphism group of a uniformly weighted complete graph. Recall that any weighted graph is given by
\begin{equation}
    G=\sum_{i\leq j\in V} g_{[i,j]}e_{[i,j]} 
\end{equation}
with $V=\{1,...,N\}$.
Then, any weighted graph can be understood as a weighting of the (unweighted) complete graph. Recalling that the action of $\SN$ in the passive picture is
\begin{equation}
\sigma\cdot G =\sum_{i\leq j\in\{1,...,N\}} g_{[i,j]} \; e_{[\sigma(i),\sigma(j)]}    
\end{equation}
it becomes clear that demanding relabeling invariance is simply the statement that it makes no difference how we choose to label the nodes of a complete graph. All points are topologically `the same', they  could all in principle be taken connected to each other, and this is to be determined by the weighting.
  
In this sense, the discrete allows us to implement background independence in a more radical way than the continuum case. In general relativity, we formally consider one of many manifolds and the diffeomorphisms that are the automorphisms of that manifold. In the discrete, we can fix as an ambient canvas a complete unweighted graph once and for all, and consider invariants under its automorphism group.\footnote{Recall that a weight of zero can be understood as removing an edge. All unweighted graphs can be understood as accordingly weighting the unweighted complete graph with ones and zeros on its edges.} This analogy in the context of continuum and discrete general relativity is given in Table \ref{tab:Analogies}.

\begin{table}[H]
\centering
\begin{tabular}{l|l}
\textbf{Continuum} & \textbf{Discrete} \\  \hline \\
Some manifold $\mathcal{M}$          & Complete graph $K$         \\
Diffeomorphisms $\phi_\mathcal{M}$          & Permutations $\sigma$  \\
Coordinates $x^{\mu}(p)$         & Labeling $l(i)$           \\
Metric components $g_{\mu\nu}(x)$                      & Weights $G_{ij}$      \\
Spacetime $(\mathcal{M},O_g)$            & Discrete spacetime $(K ,O_G)$       
\end{tabular}
\caption{Analogies between continuum and discrete general relativity notions. Diffeomorphisms are \emph{automorphisms} of a differential manifold $\mathcal{M}$ (permutations of the nodes are \emph{automorphisms} of the complete graph $K$). They act on the metric (weighted graph $G$). The metric components $g_{\mu\nu}$ in some coordinates $x^\mu(p)$ extremize the Einstein-Hilbert action (weights $G_{ij}$ in some choice of node labeling $l(i)$ extremize the Regge action). The manifold points $p$ are analogous to nodes $i$. A (discrete) spacetime geometry is the orbit $O_g$ of the metric (orbit $O_G$ of weighted graph) under diffeomorphisms of the manifold $\mathcal{M}$ (permutations of the nodes of $K$). }
\label{tab:Analogies}
\end{table}

\subsection{Continuum glocal observables}
It is not difficult to concoct continuous functionals on the space of solutions of general relativity that mimick the glocal observables we have studied. They will only be well defined on sectors of metrics where they do not diverge. Their definitions will be implicit since the `edges' will now be curves that can be defined invariantly, for instance as geodesics. 

The simplest case is the analogues of the trivial glocal observables, those that we have termed purely global (see Section \ref{sec:emptyGraph}). Take any scalar $F$ on a manifold $\mathcal{M}$ and consider the following integral
\begin{equation}
\label{eq:globalObs}
    O[g] = \int_\mathcal{M}d\omega_x  \; F(x),
\end{equation}
where $x$ is shorthand for coordinates $x^\mu$ and ${d\omega_x = dx^4 \sqrt{-g(x)}}$ is the volume form in $x^\mu$. On the collection of spacetimes where it is well defined, this object is invariant under diffeomorphisms. Observables of this form, for example, are the total volume ($F=1$) and the mean curvature ($F=R$, the Ricci scalar). On non--compact manifolds, \eqref{eq:globalObs} will diverge on most of the space of solutions. For instance, the volume is infinite for all non--compact spacetimes. If we restrict to bounded solutions on compact manifolds, or to bounded solutions and asymptotically flat spacetimes, such integrals can be well defined observables.

Functionals analogous to the glocal observables \eqref{eq:1-gloc} and \eqref{extended example} are  of the form
\begin{equation}
\label{eq:1-glocContinuum}
  O[g] = \iint_\mathcal{M} d\omega_x d\omega_y  F(x)\mathcal{F}[\gamma_{xy}]
\end{equation}
where $\gamma_{xy}$ is a geodesic connecting points $x$ and $y$, which will be uniquely defined in spacetimes without caustics, and $\mathcal{F}[\gamma_{xy}] = \int_{0}^{1}d\lambda f(\lambda)$ is evaluated along $\gamma_{xy}$ parametrized with $\lambda$. For instance, when $f = \sqrt{g_{\mu\nu}\dot{x}^\mu\dot{x}^\nu}$, $\mathcal{F}[\gamma_{xy}]$ is the proper length or proper time along the geodesic, and zero for null geodesics. A functional analogous to \eqref{eq:triangleObs}, which seeks triangular correlations, is
\begin{equation}
\label{eq:1-glocContinuum}
  O[g] \!=\!\!\! \int\!\!\!\!\int\!\!\!\!\int_\mathcal{M} \!\!\!\!d\omega_xd\omega_yd\omega_z  F(x) F(y) F(z) \mathcal{F}[\gamma_{xy}]  \mathcal{F}[\gamma_{yz}]  \mathcal{F}[\gamma_{zx}] \nonumber
\end{equation}
and so on. Vast collections of observables can be constructed in this manner. It is natural to expect---although we do not attempt to establish this here---that the geometric information encoded in observables such as the above, can be approximated arbitrarily well by invariant graph polynomials evaluated on simplicial approximations of the continuum spacetime.

It takes a moment of thought to be convinced that whenever we think of a question that does not refer to labels or coordinates, an invariant can be written that corresponds to the answer. Similarly to how we first learns about functions through polynomials, the general morals of invariant graph polynomials allow to understand more general graph invariants. For example, a counter, even in the discrete, cannot be expressed analytically. The number of edges in a weighted graph is information distributed in many invariant graph polynomials. A non-analytic function is needed to capture it in a concise formula, the characteristic function of zero $\chi_0$. Then $\frac 12N(N-1)-\sum_{i < j}\chi_0(g_{ij})$ counts the edges. To count the number of adjacent edges we write $\frac 12N(N-1)(N-2)-\sum\chi_0(g_{ij})\chi_0(g_{ik})$, and so on. An analogous example for the continuum is the two parameter family of observables
\begin{equation}
\label{eq:geodCorr}\nonumber
  O(L,R) = \int\!\!\!\!\int_\mathcal{M} \!\!\!d\omega_xd\omega_y \delta\!\left(R(x), R \right) \delta\!\left(R(y), R \right) \delta\!\left(L_{\gamma_{xy}}, L \right),
\end{equation}
where $L_{\gamma_{xy}}$ is the proper length if $x$ and $y$ are spacelike separated and zero otherwise. This counts the times two points with curvature $R$ are found at distance $L$. Similar `geodesic correlators' are used in the context of non-perturbative quantisations of gravity in \cite{Ambjorn:2012jv,Modanese:1994gh}. 

These are some examples of continuum invariants out of a plethora that can be concocted. The moral is that to write glocal invariants, we built a function encoding a kind of topological connectivity, and take the sum over the ways this structure can be embedded in the ambient canvas, a manifold for the continuum, the complete graph for the discrete. Evaluated on a configuration on which they are well defined, the glocal observables encode invariant-information for the kind of local correlation structures they correspond to.

\subsection{Incompleteness: the infinite}
The `problem of observables' of general relativity refers to at least three related but different questions: How is local information encoded in global functions? Can a complete set of observables be constructed? Can observables be found at all?  

In short, the answers are as follows. In the discrete and finite, all local correlations are completely captured with a finite set of glocal observables. Mathematically \emph{constructible} complete sets of observables cannot be found in the countable discrete case and the non--compact continuum case.\footnote{Whether complete sets of observables could be found for the compact continuum case seems to be an open question.} These anti--classification obstructions reflect \emph{limitations of the mathematical language} when dealing with the infinite, they do not reflect the physical content of the theory. Vast families of observables can be written down for the infinite case, and completeness might be possible to recover, if we appropriately restrict the state space. Let us briefly unpack these statements.

In the discrete and finite, a complete set of glocal observables is explicitly constructed by the theorem of Section \ref{sec:completeGlocal}, by taking all invariant graph polynomials for connected multigraphs which have total sum of their weights less or equal to $D \choose 2$, where $D=\frac12 N (N+1)$. In Section \ref{regge calc} we saw that this provides complete observables for discrete general relativity on simplicial complexes with a finite number of nodes. This amounts to enumerating all such multigraphs, which would take a long time, but can always be done. Alternatively, a minimal generating set can be computed using one of the known finite steps algorithms, see Section \ref{sec:minimalGeneratingSets}, which are guaranteed to successfully terminate but will also take a long time. Which procedure is faster is not known. In this sense, although impractical, a complete set of glocal observables can always be constructed. 

In the infinite, we encounter anti--classification results from descriptive set theory. In loose language, these are no--go theorems about the possibility of defining observables that distinguish isomorphic objects if we only allow certain mathematical resources, which far surpass those generally employed in mathematical and theoretical physics. For example, if we allow the target space of an observable to be an arbitrary set, and we allow the use of the axiom of choice to define a function, then we can always define a complete observable: simply map each element of an orbit to a given representative of the orbit. But this use of the axiom of choice only allows to state that something exists, the `observable' so defined does not have any practical value. In physics, we are interested in \textit{constructible} observables.

In the case of graphs with countable nodes, several anti--classification results state that no complete observable can be constructed \cite{hjorth2000classification, ros2021classificationproblemsdescriptiveset,lupini2018gamesorbitsplayobstructions}. However, we do have at our disposal many explicit functions which through the standard physicist toolbox of regularization techniques will yield observables on \emph{some} sector of infinite weighted graphs. For instance, we can write down expressions such as ${\mathcal R(X^M)= \lim_{N \rightarrow \infty} \sum_{\sigma \in \S_N} X^{\sigma \cdot M} }$ which will converge if fall-off conditions are imposed, similar to considering asymptotically flat spacetimes. This kind of pathological behavior for invariant functions should not be considered more surprising than that the function $f(x)=x^2$ is not integrable over the reals because the Riemann integral does not converge. 

That a complete set of observables can not be constructed for full general relativity, even for the vacuum sector, was shown recently in \cite{Panagiotopoulos_2023}. In this view, the `problem of observables' is that the action of diffeomorphisms on the space of metrics is ergodic, each orbit is `everywhere' in the space of solutions, and it is impossible to tell them apart. As discussed above, this is not a phenomenon inherent to the continuum.  The importance of \cite{Panagiotopoulos_2023} was to give a first such result for general relativity, in fact, it demonstrates a milder form of anti--classification obstruction than what has been already proven for countable graphs. Anti--classification obstructions in general relativity are probably more severe than those for countable graphs but more difficult to prove. Nevertheless, as discussed in the previous section, by restricting on appropriate sectors, vast sets of observables can be written down also for the continuum. 

In loose terms, infinities have the habit to cause mathematical problems whenever we build them in a model of physics, unless we tame them somehow. Searching for generic complete sets of observables on infinite structures is to request mathematics to do something it cannot. If completeness is desired in the infinite case, this might be possible but the space of solutions of the theory will need to be significantly restricted.

The lessons learned in this work suggest a strategy to recover completeness in the infinite case. Incompleteness does not seem to imply that there is `some other kind of information' that glocal observables cannot capture. Rather, the issue seems to be that glocal observables diverge or are otherwise ill-defined if the sector of solutions is not appropriately restricted. Therefore, demanding that glocal observables are well-defined may allow to identify appropriate sectors for the infinite case so that completeness is recovered. 

\section*{Summary}

In this work, we demonstrated how the main conceptual insight of general relativity, background independence, can be implemented fully for theories defined on graphs. The lack of background, with respect to which to define localization a priori, implies that the invariants will be global functions of the entire spacetime. We have seen that (i) this is directly analogous to the task of completely determining the label-invariant information held on a graph: the information encoded on a weighted graph modulo graph isomoprhisms (ii) there is no tension with the fact that glocal information needs to be captured globally, this can be done completely through the kind of invariants we have termed glocal observables.

We argued that general covariance in a discrete setting corresponds to the invariance under all permutations of the node labels of a weighted graph, implemented so that the adjacency \emph{relations} among weights are preserved.  Then, weighted graphs seen as a vector space provide a natural arena to accommodate background independent theories of physics, with the weights specifying a discrete spacetime metric. Before specifying any configuration, all points are the same, exchangeable, and may in principle be all connected to each other. Locality, both in the sense of topological connectivity and in terms of distance, are specified at the same time, by the specification of a weighted graph. 

Each glocal observable is an invariant average that captures information about a type of local correlation. There is no contradiction between locality and the global character of these observables. For instance, it may be that a certain type of local correlation only appears once in a weighted graph, a unique local feature. The observable will search the entire graph and return this local information.

Complete sets of glocal observables can be built out of algebraic invariants. Generating sets are constructed out of graph invariants. Then, we have seen that any algebraic invariant can be algebraically reduced to graph invariants associated with connected subgraph structures---glocal observables. In addition, it is guaranteed that a finite and complete set of glocal observables exists. Their construction is algorithmic and terminates in a finite number of steps. These complete sets of glocal observables capture \emph{all} label independent information on the graph, and generate all other observables. In this precise sense, all geometrical information on a weighted graph \emph{is} glocal information.

We have seen how the formalism can be applied to background independent theories of discrete spacetime. We pointed out that to the degree Lorentzian Regge calculus approximates general relativity, the structures described above provide a constructible resolution of the problem of observables. We have also seen how to construct a complete set of invariants for the quantum case of spin networks, the kinematical setup of loop quantum gravity. We constructed a symmetrised state space as the span of equal superpositions of labeled spin networks in an orbit---unlabelled spin networks---the natural label-independent quantum states arising from the formalism. This basis yields a constructible background-independent definition of quantum operations that change graph connectivity, and can be applied to the study of general transitions between quantum geometries. We proposed that this native--to--the--discrete procedure is how the spatial diffeomorphism constraint of general relativity can be implemented in loop quantum gravity.

This work sketches deep connections between questions of fundamental importance for the theory of spacetime and the theory of computation: the problem of observables and the graph isomorphism problem. The foundational reasoning that leads to both questions is common, as both concern the encoding of local information independently of a background structure. In computer science, this allows to distinguish non--isomorphic graphs, in gravitational physics to distinguish different spacetime geometries. We hope to have adequately made the case to further uncover the links between these traditional lines of investigation.

\section*{outlook}

We close with speculative perspectives on how this work may open paths for computing methods to be used for the characterization of classical and quantum spacetimes.

Parallels with recent influential ideas in deep learning are intriguing. Permutation invariance underlies the design and universality of certain architectures of machine learning models  \cite{zaheer2018deepsets, fereydounian2022functionsgraphneuralnetworks, kimura2024permutationinvariantneuralnetworks}, in order to avoid learning fictitious dependencies on arbitrary labels used to describe the input. These models are known to process data efficiently when working with unordered or graph structured data. Permutation-invariant deep learning algorithms could be trained on label--dependent solutions of the equations of motion, and learn the correlation structures between the permutation-invariant observables. This might allow for a formulation of the laws of physics for discrete spacetime in terms of glocal observables.

Explicitly finding a complete set of observables for the space of all weighted graphs up to some number of nodes is tantamount to solving the graph isomorphism problem, and is expected to remain computationally costly. However, the graph isomorphism problem is known to scale polynomialy in time for bounded-degree graphs \cite{luks1982isomorphism}. Therefore, an interesting question is whether complete sets of observables with cardinality considerably smaller than suggested by the theorem of Section \ref{theorem} can be constructed for discrete general relativity, if we restrict to simplicial complexes with bounded-valence skeletons. The same question is posed for the kinematical space of loop quantum gravity, since standard formulations restrict to only four--valent spin networks. 

The relation of permutation invariance and entanglement is also intriguing. The invariant graph states of spin networks naturally arising from our formalism are highly entangled superpositions of the elements in an orbit of permutations. This is the natural quantum way to erase label dependence, which can be thought of as assigning to each point an equal superposition of all labels, in complete analogy to what happens with bosonic systems. 

It is then natural to extend this work to invariants under quantum changes of coordinates \cite{Kabel:2024lzr, Hardy2018construction,Hardy2019quantum,Hardy2016operational,delaHamette2022quantum,MekonnenGalleyMueller2025QuantumPermutationsParastatistics}: preliminary results have shown that the formalism developed here yields a good definition of invariant entanglement for quantum graph states \cite{EmilMitrakos}. A related avenue of investigation is to extend our analysis to invariants under the quantum group of permutations \cite{Weber2023QuantumPermutationMatrices,atserias2019quantum,ofekMarios}.

Finally, it may be possible to study quantum geometry with permutation invariant quantum computing codes, which are seeing a recent surge of interest \cite{Jordan2009_PermutationalQC,mansky2025solvinggraphproblemsusing,Heredge2024PermutationInvariant,ZhengEtAl2022_PQCplus, Anschuetz_2023,PhysRevA.100.052317}. We also note that a most well-known benchmark for quantum platforms, boson sampling  \cite{BulmerEtAl2021, ChabaudDeshpandeMehraban2022,DeshpandeEtAl2021,RahimiKeshariEtAl2015,ZhongEtAl2020,HamiltonEtAl2017,AaronsonArkhipov2013,TillmannDakić2013,Wei_2010}, is the calculation of the permanent, one of the observables we have called here purely global. It is particularly interesting to look into these possible connections in relation with bosonic formulations of quantum geometry \cite{Bianchi:2016hmk,
Freidel:2010tt,
Freidel:2010bw, Oriti:2013aqa}.  \\ \bigskip

\begin{acknowledgements}
We thank Jeremy Butterfield and Aristotelis Panagiotopoulos for insightful discussions on the importance of permutation invariance and its relation to the problem of observables.  

 We acknowledge support of the ID~\#~62312 grant from the John Templeton Foundation, as part of the ``Quantum Information Structure of Spacetime (QISS)'' project (\hyperlink{http://www.qiss.fr}{qiss.fr}). Emil Broukal acknowledges support of the grant ID~\#~63132 from the John Templeton Foundation, as recipient of an Enrico Fermi Fellowship awarded through the Center for SpaceTime and the Quantum. The opinions expressed in this publication are those of the author(s) and do not necessarily reflect the views of the respective funding body. Eugenio Bianchi acknowledges support from the National Science Foundation, Grants No. PHY-2207851 and PHY-2513194.
\end{acknowledgements}

\raggedbottom

\bibliography{refs}

\bigskip
\onecolumngrid

\appendix

\section{Definitions}\label{graph defs}

\subsection{From abstract graphs to adjacency matrices of weighted graphs}

{\bf Abstract graph.} A (simple undirected) graph on $N$ nodes is given by a pair $G=(V,E)$, where $V$ is a finite set consisting of $N$ elements, called nodes (or vertices) and $E$ is a set of unordered pairs of nodes, whose elements are called edges.\\

{\bf Abstract graph isomorphism.} Let $G,H$ be two graphs and let $V(G),V(H)$ and $E(G),E(H)$ denote their vertex and edge sets respectively. An isomorphism between $G$ and $H$ is a bijection $\varphi:V(G)\xrightarrow[]{}V(H)$ such that $\{u,v\}\in E(G) \Leftrightarrow\{\varphi(u),\varphi(v)\}\in E(H)$, i.e. two nodes $u,v$ in $G$ are adjacent, if and only if $\varphi(u),\varphi(v)$ are adjacent in $H$. If an isomorphism exists between two graphs $G,H$, they are called isomorphic and denoted as $G\cong H$.\\

{\bf Labeling and labeled graphs.} A labeling of a graph $G=(V,E)$ on $N$ nodes is an enumeration of the set of nodes $V$ by the consecutive integers $\Tilde{V}=\{1,...,N\}$, i.e. an bijection $l:V\xrightarrow[]{}\Tilde{V}$. Note that this induces an enumeration of the edges as well: $\Tilde{E}=\big\{\{i,j\} ~|~ \{l^{-1}(i),l^{-1}(j)\}\in E \big\}$. We call the pair of enumerated nodes and edges $\Gamma=(\Tilde{V},\Tilde{E})$ a labeled graph. Note that giving a labeling involves a choice, namely the order in which one enumerates the graphs nodes. Choosing a specific enumeration is of course pure convention.\\

{\bf Adjacency matrix of a labeled graph.} Let $\Gamma=(\{1,...,N\},E)$ be a labeled graph. We define its adjacency matrix $A_{\Gamma}$ as
\begin{align}
    (A_{\Gamma})_{ij}=\begin{cases}
        1\text{, if }\{i,j\}\in E\\
        0\text{, else}
    \end{cases}
\end{align}
From this definition it follows that $(A_{\Gamma})_{ij}=(A_{\Gamma})_{ji}$ and $(A_{\Gamma})_{ii}=0$. Note that the assignment of the rows and columns of $A_{\Gamma}$ to the nodes of $\Gamma$ is only possible unambiguously because it is labeled.\\

{\bf Weighted graph.} Let $G=(V,E)$ be a graph. We call a function
\begin{equation}
\begin{aligned}
    w_n:~V&\longrightarrow\mathbb{R}\\
    v&\longmapsto w_n(v)
\end{aligned}
\end{equation}
a \textit{node weight function} and
\begin{equation}
\begin{aligned}
    w_e:~~~~~E&\longrightarrow\mathbb{R}\setminus\{0\}\\
    \{v,u\}&\longmapsto w_e(\{v,u\})
\end{aligned}
\end{equation}
an \textit{edge weight function}. Finally, we call $G_w=(V,E,w_n,w_e)$ a weighted graph.\\

{\bf Labeled weighted graph.} We call a pair $G=(\Gamma,w)$, where $\Gamma$ is a labeled graph on $N$ nodes and $w=(w_n,w_e)$ with $w_n (w_e)$ a node (edge) weight function on $\Gamma$, a labeled weighted graph.\\

{\bf Adjacency matrix of a labeled weighted graph} Let $G=(\Gamma,w)$ be a labeled weighted graph. We define its weighted adjacency matrix $A_G$ as
\begin{align}
    (A_G)_{ij}=\begin{cases}
        w_n(i)\text{, if } i=j\\
        w_e(\{i,j\})\text{, if }\{i,j\}\in E\\
        0,\text{ else}
    \end{cases}
\end{align}
Again, it follows from this definition that $(A_G)_{ij}=(A_G)_{ji}$. Note that a labeled weighted graph $G$ is fully specified by its weighted adjacency matrix $A_G$: Let $A\in\text{Sym}(\mathbb{R},N)$, that is a symmetric $N\times N$ matrix with real entries, and define $E=\{\{i,j\} ~|~ i< j \in \{1,...,N\}\text{ and }A_{ij}\neq 0\}$, $w_n:\{1,...,N\}\rightarrow\mathbb{R},~ w_n(i)=A_{ii}$, and $w_e:E\rightarrow\mathbb{R},~ w_e(\{i,j\})=A_{ij}$. Then $(\{1,...,N\},E,w_n,w_e)$ is a labeled weighted graph. \\

{\bf The vector space of labeled weighted graphs} We denote by $\mathcal G_N$ the space of all labeled weighted graphs over $N$ nodes. This can be presented as a vector space, see Section~\ref{ref:vectorSpace}. As a vector space, $\mathcal G_N$ is isomorphic to the space of $N\times N$ adjacency matrices.\\

{\bf Non-trivial connected component} Let $\Gamma=(\{1,..,N\},E)$ be a labeled graph. A labeled graph is called \textit{connected} if their exists a path along edges from every node to every other. Let $c=(V_c,E_c)$ where $V_c\subset\{1,...,N\}$ and ${E_c=\{\{i,j\}~|~ \{i,j\}\in E \text{ and } i,j\in V_c\}}$. $c$ is then called a \textit{subgraph} of $\Gamma$. Furthermore, we call $c$ a \textit{non-trivial connected component} of $\Gamma$, if $c$ has at least one edge and is a maximal connected subgraph, meaning that including any additional node and the corresponding edges would make it non-connected.\\

{\bf Quasi connected graph} Let $G=(\Gamma,w)$ be a labeled weighted graph. We call $G$ \textit{quasi connected}, if $\Gamma$ has exactly one non-trivial connected component, or $G$ has no edges and exactly one weighted node.

\subsection{Relabelings and permutations}

{\bf Action of $\S_N$ on labeling induces action on adjacency matrices.} Let $A_G\in \text{Sym}(\mathbb{R},N)$ be a weighted adjacency matrix of a labeled weighted graph $G$. A relabeling of $G$ is a transformation of its node set by an element $\sigma$ of $\S_N$, i.e. $n\mapsto \sigma(n)$, which naturally induces a transformation of its edge set: $\sigma\cdot E=\big\{\{\sigma(i),\sigma(j)\}~|~ \{i,j\}\in E\big\}$. Note that this extension preserves the adjacency relations of the graph. This in turn, induces an action of $\S_N$ on the weighted adjacency matrix $A_G$, which transforms into a new weighted adjacency matrix $\sigma\cdot A_G$, satisfying
\begin{align}
    (\sigma\cdot A_G)_{\sigma(i)\sigma(j)}=(A_G)_{ij}.
\end{align}

{\bf Graph isomorphism and action of $\S_N$} Let $G=(\{1,..,N\},E^G,w^G_n,w^G_e)$ and $H=(\{1,...,N\},E^H,w^H_n,w^H_e)$ be two labeled weighted graphs. We say that $G$ and $H$ are isomorphic as weighted graphs, if there exists an isomorphism $\varphi$ between the two labeled graphs $G$ and $H$ without weights and in addition
\begin{align}
    w^G_n\circ\varphi^{-1}=w^H_n\\
    w^G_e\circ\varphi^{-1}=w^H_e
\end{align}
holds. It is a well known fact from graph theory that two simple undirected graphs $G,H$ are isomorphic, if and only if there exists a permutation $\sigma\in\S_N$ such that $\sigma\cdot A_G=A_H$, or equivalently $(A_G)_{ij}=(A_H)_{\sigma(i)\sigma(j)}$. This alternative characterization of isomorphism naturally carries over to the case of labeled weighted graphs word for word: two labeled weighted graphs $G,H$ are isomorphic if and only if there exists a permutation $\sigma\in\S_N$ such that $(A_G)_{ij}=(A_H)_{\sigma(i)\sigma(j)}$ holds, where $A_G,A_H$ denote their weighted adjacency matrices respectively.\\

\section{Reynolds operator and invariant graph polynomials}
\label{sec:propertiesReynolds}
{\bf Properties of the Reynolds operator} Let $\mathcal{R}:\mathbb{R}[\mathcal{G}_N]\longrightarrow\mathcal{I}^N$ be the operator as defined in \eqref{eq:reynolds_def}. Then $\mathcal{R}$ has the following properties:
\begin{itemize}
    \item $\mathcal{R}$ is linear.
    \item $\mathcal{R}(1)=1$, where $1$ denotes the constant polynomial with value 1.
    \item $\mathcal{R}$ satisfies the Reynolds property: For any $p,q\in\mathbb{R}[\mathcal{G}_N]$, $\mathcal{R}(\mathcal{R}(p)q)=\mathcal{R}(p)\mathcal{R}(q)$.
    \item $\mathcal{R}$ acts trivially on $\mathcal{I}^N\subset\mathbb{R}[\mathcal{G}_N]$.
\end{itemize}
Linearity follows from the fact that the action of $\S_N$ on $\mathcal{G}_N$ is defined to be linear. The second property can be seen by:
\begin{align}
    \mathcal{R}(1)=\frac{1}{|\S_N|}\sum_{\sigma\in\S_N}\sigma\cdot 1=\frac{1}{|\S_N|}\sum_{\sigma\in\S_N}=1.
\end{align}
Similarly, the Reynolds property can be shown as follows:
\begin{align}
    \mathcal{R}(\mathcal{R}(p)q)&=\frac{1}{|\S_N|}\sum_{\sigma\in\S_N}\sigma\cdot\big(\mathcal{R}(p)q\big)=\frac{1}{|\S_N|^2}\sum_{\sigma\in\S_N}\sum_{\pi\in\S_N}\sigma\cdot\big((\pi\cdot p)q\big)=\frac{1}{|\S_N|^2}\sum_{\sigma\in\S_N}\sum_{\pi\in\S_N}(\sigma\cdot\pi\cdot p)(\sigma\cdot q)\\
    &=\frac{1}{|\S_N|^2}\sum_{\sigma\in\S_N}\sum_{\pi'\in\S_N}(\pi'\cdot p)(\sigma\cdot q)=\frac{1}{|\S_N|^2}\sum_{\pi'\in\S_N}(\pi'\cdot p)\sum_{\sigma\in\S_N}(\sigma\cdot q)=\mathcal{R}(p)\mathcal{R}(q),
\end{align}
where we used that $\sigma\cdot(pq)=(\sigma\cdot p)(\sigma\cdot q)$ for any $\sigma\in\S_N,~p,q\in\mathbb{R}(\mathcal{G}_N)$. Note that from the second and third properties, it follows that $\mathcal{R}$ is idempotent:
\begin{equation}
    \mathcal{R}(\mathcal{R}(p))=\mathcal{R}(\mathcal{R}(p)1)=\mathcal{R}(p)\mathcal{R}(1)=\mathcal{R}(p).
\end{equation}
The final property can be shown in the following way: Let $p\in\mathcal{I}^N$. Then
\begin{equation}
    \mathcal{R}(p)=\frac{1}{|\S_N|}\sum_{\sigma\in\S_N}\sigma\cdot p=\frac{1}{|\S_N|}\sum_{\sigma\in\S_N} p=p.
\end{equation}
\\

{\bf Properties of invariant graph polynomials} We will now show some properties of invariant graph polynomials, i.e. invariant polynomials obtained by acting on graph monomials with the Reynolds operator. Consider two weighted graphs $G_1,G_2\in\mathcal{G}_N$ that are isomorphic, i.e. there exists a $\sigma\in\S_N$ such that $\sigma\cdot G_1=G_2$. Then it straightforwardly follows that $\mathcal{O}_{G_1}=\mathcal{O}_{G_2}$, that is, their orbits coincide. Thus, trivially, $\mathcal{R}(X^{G_1})=\mathcal{R}(X^{G_2})$, as
\begin{align}
    \mathcal{R}(X^{G_1})=\frac{1}{|\mathcal{O}_{G_1}|}\sum_{H\in\mathcal{O}_{G_1}}X^H=\frac{1}{|\mathcal{O}_{G_2}|}\sum_{H\in\mathcal{O}_{G_2}}X^H=\mathcal{R}(X^{G_2}).
\end{align}
Furthermore, recall the formula for the product of two invariant graph polynomials:
\begin{equation}\label{eq:Rexplcit}
    \mathcal{R}(X^{G_1})\mathcal{R}(X^{G_2})=\frac{1}{|\mathcal{O}_{G_2}|}\sum_{H\in\mathcal{O}_{G_2}}\mathcal{R}(X^{G_1+H}).
\end{equation}
This can be shown as follows:

\begin{align}
    \mathcal{R}(X^{G_1})\mathcal{R}(X^{G_2})=\mathcal{R}\big(X^{G_1}\mathcal{R}(X^{G_2})\big)=\mathcal{R}\bigg(\frac{1}{|\mathcal{O}_{G_2}|}\sum_{H\in \mathcal{O}_{G_2}}X^{G_1+H}\bigg)=\frac{1}{|\mathcal{O}_{G_2}|}\sum_{H\in \mathcal{O}_{G_2}}\mathcal{R}(X^{G_1+H}),
\end{align}

\noindent by the Reynolds property and linearity. Equation~\eqref{eq:Rexplcit} shows that the product of any two invariant graph polynomials can be calculated by considering every possible way to superimpose the two underlying graphs and sum over all possible invariant graph polynomials obtained by applying the Reynolds operator to these superimpositions. This can be used to express any invariant graph polynomial in terms of a set of primitive ones, which is captured in the following theorem.

\begin{theorem*}
    Let $\mathcal{C}$ be the set of all invariant polynomials of the form $\mathcal{R}(X^{M_1})\cdot\cdot\cdot\mathcal{R}(X^{M_k})$, where each $M_i$ is a quasi connected multigraph with $n_i$ non-isolated nodes and sum of weights $W(M_{i})$ such that $n_1+...+n_k\leq N$ and $\sum_{i=1}^kW(M_i)=d$. Then $\mathcal{C}$ is a vector space basis for the homogeneous component $\mathcal{I}^N_d$.
\end{theorem*}
\begin{proof}
    Let $p\in\mathcal{I}^N_d$. As $p$ is a homogeneous polynomial of degree $d$, it can be written as a linear combination of monomials $q_i$ of degree $d$, i.e. $p=\sum_i\alpha_iq_i$. Then each $q_i=X^{M_{q_i}}$ where $M_{q_i}$ is the multigraph on $N$ nodes with weights given by the exponents in $q_i$. Note that $W(M_{q_i})=d$. As $p$ is invariant, we have that
    \begin{equation}
        p=\mathcal{R}(p)=\sum_i \alpha_i \mathcal{R}(X^{M_{q_i}}).
    \end{equation}
    To complete the proof, we show that for any multigraph $M$, $\mathcal R(X^M)$ can be written as a linear combination of elements in $\mathcal C.$
    Let $M$ be a multigraph on $N$ nodes with $W(M)=d$ with $k$ connected components $c_1,\dots,c_k$. Assume $k>1$
and let $\Bar{c}_1,\dots,\Bar{c}_k\in \mathcal G_N$ be the graphs obtained by taking each connected component $c_i$ of $M$ and adding to it $N-n_i$ isolated nodes with weight $0$, where $n_i$ denotes the number of nodes of $c_i$. By construction, each $\Bar{c}_i$ has exactly one connected component and they satisfy $n_1+\cdots+n_k\leq N$ and $\sum_i W(\Bar{c}_i)=d$. By the properties of invariant graph polynomials, it then follows that
\begin{equation}
    \mathcal{R}(X^{\Bar{c}_1})\cdot\cdot\cdot\mathcal{R}(X^{\Bar{c}_k})=\mathcal{R}(X^M)+\sum_{i}\mathcal{R}(X^{h_i}),
\end{equation}
where each $h_i$ is a multigraph with strictly less then $k$ non-trivial connected components. By induction on the number of connected components $k$, it follows that $\mathcal{R}(X^M)$ is a linear combination of products of the form $\mathcal{R}(X^{\Bar{c}_1})\cdot\cdot\cdot\mathcal{R}(X^{\Bar{c}_k})$. Finally, note that there is nothing to proof for $k=0,1$, as then $M$ already only has at most one non-trivial connected component.
\end{proof}

{\bf Proof of theorem \ref{theorem}} \\
We now prove a slightly stronger version of the theorem in the main text, which then follows as a corollary. 
\begin{theorem*}
    Let
    $I^N_C = \left\{ O^M ~\middle|~ W(M)\leq \beta(\mathcal I^N)\right\},$
    where $M$ ranges over all quasi-connected multigraphs and $W(M)$ is the sum of the weights on $M$. Then $I^N_C$ generates the full invariant algebra $\mathcal{I}^N$.
\end{theorem*}
\begin{proof}
By Hilbert's finiteness theorem, the existence of a minimal generating set for $\mathcal{I}^N$ is guaranteed. In addition, $\beta(\mathcal{I}^N)$ is well defined. Note that by enlarging a generating set, the generating property will not be lost. We will now show that any minimal generating set can be enlarged to $I^N_C$. Consider an arbitrary minimal generating set of $\mathcal{I}^N$, and denote it by $I$. Decompose each polynomial contained in $I$ into its homogeneous components. Denote the new set of homogeneous polynomials obtained this way by $\Bar{I}$. One thus has a set of polynomials where each belongs to some sector $\mathcal{I}^N_d$, where $d=0,\cdots,\beta(\mathcal{I}^N)$. Recall that by the above theorem, one can express each polynomial in $\Bar{I}$ as a sum of terms of the form $\mathcal{R}(X^{M_1})\cdot\cdot\cdot\mathcal{R}(X^{M_k})$, where each $M_i$ is a quasi connected multigraph such that $\sum_{i=1}^KW(M_i)=d$, where $d$ is the degree of the homogeneous polynomial one wants to write in this basis. Next, express each element of $\Bar{I}$ in such a basis and add every basis element of every $\mathcal{I}^N_d$, $d=0,...,\beta(\mathcal{I}^N)$ that is not yet contained in $\Bar{I}$, creating a larger set $\tilde{I}$ that is no longer minimal. Finally, simply decompose each linear combination of products in $\tilde{I}$ into its constituents, which are all the invariant graph polynomials $\mathcal{R}(X^{M_i})$, where $M_i$ is a quasi connected multigraph with $W(M_i)\leq\beta(\mathcal{I}^N)$. This then gives exactly $I^N_C$.
\end{proof}
The above two theorems are a generalisation of Proposition 2.1 in \cite{thiéry2008}.
\section{Complete observables}
\label{sec:FiniteGeneratingSetsCompleteObservables}
{\bf Finite generating sets are complete observables} Let $\{O_1,...,O_r\}$ be a finite collection of observables. We call $\{O_1,...,O_r\}$ a complete set of observables, if and only if for any $A,B\in\mathcal{G}_N$,
\begin{align}
    A\cong B \Leftrightarrow O_i(A)=O_i(B)~\forall i=1,...,r
\end{align}
holds. The above condition means that $\{O_1,...,O_r\}$ separates orbits of $\mathcal{G}_N$ under $\S_N$. A finite generating set for $\mathcal{I}^N$ will also constitute a complete set of observables on $\mathcal{G}_N$. We will now give the proof of this (for more details, see \cite[Theorem 10]{inv2}).\\

\begin{proof}
    
Let $\{i_1,...,i_r\}$ be a finite generating set for $\mathcal{I}^N$ and consider two graphs $G_1,G_2\in\mathcal{G}_N$. If $G_1\cong G_2$, then it trivially follows that $i_k(G_1)=i_k(G_2),~ \forall k=1,...,r$. So we assume that $G_1$ and $G_2$ are not isomorphic, which implies that $\mathcal{O}_{G_1}\cap\mathcal{O}_{G_2}=\emptyset$. We will now construct an invariant $g\in \mathcal{I}^N$ such that $g(G_1)\neq g(G_2)$. First, define $\Omega=\mathcal{O}_{G_1}\cup\mathcal{O}_{G_2}\setminus\{G_1\}$. Note that by the following identifications $\mathcal{G}_N\cong\text{Sym}(N,\mathbb{R})\cong\mathbb{R}^{\frac{N(N+1)}{2}}$, we can identify $\Omega$ with a finite set of points in $\mathbb{R}^{\frac{N(N+1)}{2}}$. Since every finite set of points in $\mathbb{R}^{\frac{N(N+1)}{2}}$ is an affine variety, there exists a finite set of polynomials $F=\{p\in\mathbb R[\mathcal G_N]\}$ such that $\Omega$ is the set of simultaneous roots of all polynomials in $F$, \textit{i.e.} $p[G]=0$ for all $G\in\Omega$ and all $p\in F$. Since $G_1\not\in \Omega$ there is at least one $f\in F$ such that $f(G_1)\neq0.$

Let now $g=\mathcal{R}(f)$. Then
\begin{align*}
    &g(G_2)=\frac{1}{|\S_N|}\sum_{\sigma\in\S_N}f(\sigma\cdot G_2)=0,~\text{as }f\text{ vanishes on all elements in }\mathcal{O}_{G_2},\\
    &g(G_1)=\frac{1}{|\S_N|}\sum_{\sigma\in\S_N}f(\sigma\cdot G_1)=\frac{|\text{Aut}(G_1)|}{|\S_N|}f(G_1)\neq 0,
\end{align*}
which follows from the fact that $f(\sigma\cdot G_1)=f(G_1)\neq 0 \Leftrightarrow \sigma\cdot G_1=G_1$. We have thus constructed an invariant that seperates $G_1$ and $G_2$. Since $\{i_1,...i_r\}$ are a generating set, we have that $g=h(i_1,...,i_r)$ for some polynomial $h$. The fact that $g(G_1)\neq g(G_2)$ implies that there must be at least one $i_k$ such that $i_k(G_1)\neq i_k(G_2)$, which completes the proof.
\end{proof}

\section{Minimal generating sets for $\mathcal{I}_3$ and $\mathcal{I}_4$}\label{appendix: minimal generating set}

\paragraph*{\textit{Minimal generating set for $n=3$ and $n=4$}.}
Using an implementation of King's algorithm in the computer algebra program SageMath presented in \cite{masters}, we explicitly computed a minimal generating set for the algebras of invariants $\mathcal{I}^3$ and $\mathcal{I}^4$ of weighted graphs with 3 and 4 nodes. The minimal generating set for $\mathcal{I}^3$ consists of 9 polynomials with degrees $(2,3,4)$, where the position in the tuple indicates the degree of the number of polynomials, see Figure \eqref{fig:mingenfor n=3}. Note that therefore $\beta(\mathcal{I}^3)=3$. We suppress use of the multiset index notation for conciseness, that is, $x_{ij}=x_{[i,j]}$.

\begin{figure}[H]
    \centering
    \includegraphics[width=0.7\linewidth]{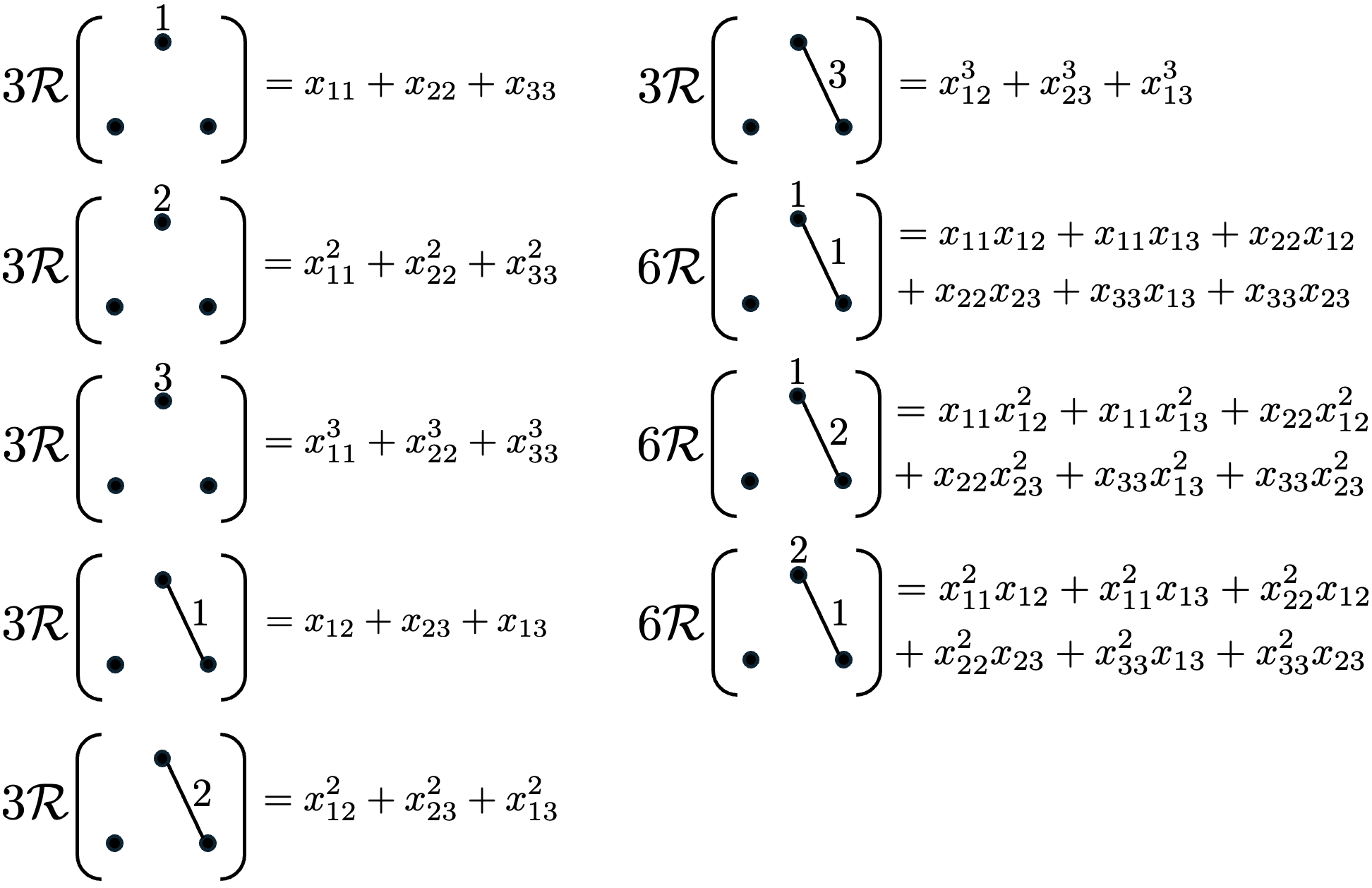}
    \caption{Minimal generating set for $\mathcal{I}^3$}
    \label{fig:mingenfor n=3}
\end{figure}

For $\mathcal{I}^4$, we found a minimal generating set consisting of 31 invariant polynomials with smallest degree bound $\beta(\mathcal{I}^4)=5$. The number of polynomials per degree is $(2,4,8,10,7)$, where the position in the tuple again corresponds to the degree of the polynomials. To increase readability, we denote the monomials corresponding to the node weights by ${A=x_{[1,1]}}, {B=x_{[2,2]}}, {C=x_{[3,3]}}, {D=x_{[4,4]}}$. The minimal generating set is then given by:

\begin{align*}
    &Q_1=A+B+C+D,~~ Q_2=x_{12}+x_{13}+x_{14}+x_{23}+x_{24}+x_{34}\\
    &P_1=A^2+B^2+C^2+D^2,~~ P_2=x_{12}^2+x_{13}^2+x_{14}^2+x_{23}^2+x_{24}^2+x_{34}^2,\\
&P_3=A(x_{12}+x_{13}+x_{14})+B(x_{12}+x_{23}+x_{24})+C(x_{13}+x_{23}+x_{34})+D(x_{14}+x_{24}+x_{34}),\\
&P_4=x_{12}x_{13}+x_{12}x_{14}+x_{13}x_{14}+x_{12}x_{23}+x_{13}x_{23}+x_{12}x_{24}+x_{14}x_{24}+x_{23}x_{24}+x_{13}x_{34}+x_{14}x_{34}+x_{23}x_{34}+x_{24}x_{34}\\
&R_1=A^3+B^3+C^3+D^3,~~ R_2=x_{12}^3+x_{13}^3+x_{14}^4+x_{23}^3+x_{23}^3+x_{24}^3+x_{34}^3\\
&R_3=A^2(x_{12}+x_{13}+x_{14})+B^2(x_{12}+x_{23}+x_{24})+C^2(x_{13}+x_{23}+x_{34})+D^2(x_{14}+x_{24}+x_{34})\\
&R_4=A(x_{12}^2+x_{13}^2+x_{14}^2)+B(x_{12}^2+x_{23}^2+x_{24}^2)+C(x_{13}^2+x_{23}^2+x_{34}^2)+D(x_{14}^2+x_{24}^2+x_{34}^2)\\
&R_5=ABx_{12}+ACx_{13}+ADx_{14}+BCx_{23}+BDx_{24}+CDx_{34}\\
&R_6=A(x_{12}x_{13}+x_{12}x_{14}+x_{13}x_{14})+B(x_{12}x_{23}+x_{12}x_{24}+x_{23}x_{24})+C(x_{13}x_{34}+x_{23}x_{34}+x_{13}x_{23})\\
&~~~~+D(x_{14}x_{24}+x_{14}x_{34}+x_{24}x_{34}),\\
&R_7=x_{12}x_{13}x_{23}+x_{12}x_{14}x_{24}+x_{13}x_{14}x_{34}+x_{23}x_{24}x_{34}\\
&R_8=x_{12}^2x_{13}+x_{12}^2x_{14}+x_{12}^2x_{23}+x_{12}^2x_{24}+x_{13}^2x_{23}+x_{13}^2x_{34}+x_{14}^2x_{24}+x_{14}^2x_{34}+x_{23}^2x_{24}+x_{23}^2x_{34}+x_{34}^2x_{23}+x_{34}^2x_{24}\\
&~~~~+(\text{exponents switched})\\
&S_1=A^4+B^4+C^4+D^4,~~ S_2=x_{12}^4+x_{13}^4+x_{14}^4+x_{23}^4+x_{24}^4+x_{34}^4\\
&S_3=A^3(x_{12}+x_{13}+x_{14})+B^3(x_{12}+x_{23}+x_{24})+C^3(x_{13}+x_{23}+x_{34})+D^3(x_{14}+x_{24}+x_{34})\\
&S_4=A^2(x_{12}^2+x_{13}^2+x_{14}^2)+B^2(x_{12}^2+x_{23}^2+x_{24}^2)+C^2(x_{13}^2+x_{23}^2+x_{34}^2)+D^2(x_{14}^2+x_{24}^2+x_{34}^2)\\
&S_5=A^2(x_{12}x_{13}+x_{12}x_{14}+x_{13}x_{14})+B^2(x_{12}x_{23}+x_{12}x_{24}+x_{23}x_{24})+C^2(x_{13}x_{34}+x_{23}x_{34}+x_{13}x_{23})\\
&~~~~+D^2(x_{14}x_{24}+x_{14}x_{34}+x_{24}x_{34})\\
&S_6=A(x_{12}^3+x_{13}^3+x_{14}^3)+B(x_{12}^3+x_{23}^3+x_{24}^3)+C(x_{13}^3+x_{23}^3+x_{34}^3)+D(x_{14}^3+x_{24}^3+x_{34}^3)\\
&S_7=ABx_{12}^2+ACx_{13}^2+ADx_{14}^2+BCx_{23}^2+BDx_{24}^2+CDx_{34}^2\\
&S_8=A(x_{12}^2x_{13}+x_{12}^2x_{14}+x_{13}^2x_{14}+\text{(exp. sw.)})+B(x_{23}^2x_{12}+x_{24}^2x_{12}+x_{23}^2x_{24}+(\text{exp. sw.}))\\
&~~~~+C(x_{34}^2x_{13}+x_{34}^2x_{23}+x_{23}^2x_{13}+(\text{exp. sw.}))+D(x_{24}^2x_{14}+x_{34}^2x_{14}+x_{34}^2x_{24}+(\text{exp. sw.}))\\
&S_9=A(x_{12}^2x_{23}+x_{13}^2x_{23}+x_{12}^2x_{24}+x_{14}^2x_{24}+x_{13}^2x_{34}+x_{14}^2x_{34})+B(x_{12}^2x_{13}+x_{12}^2x_{14}+x_{23}^2x_{13}+x_{23}^2x_{34}+x_{24}^2x_{14}+x_{24}^2x_{34})\\
&~~~~+C(x_{13}^2x_{12}+x_{13}^2x_{14}+x_{23}^2x_{12}+x_{23}^2x_{24}+x_{34}^2x_{14}+x_{34}^2x_{24})+D(x_{14}^2x_{12}+x_{14}^2x_{13}+x_{24}^2x_{12}+x_{24}^2x_{23}+x_{34}^2x_{13}+x_{34}^2x_{23})\\
&S_{10}=x_{12}^3x_{13}+x_{12}^3x_{14}+x_{12}^3x_{23}+x_{12}^3x_{24}+x_{13}^3x_{23}+x_{13}^3x_{34}+x_{14}^3x_{24}+x_{14}^3x_{34}+x_{23}^3x_{24}+x_{23}^3x_{34}+x_{34}^3x_{23}+x_{34}^3x_{24}\\
&~~~~+(\text{exponents switched})\\
&T_1=x_{12}^5+x_{13}^5+x_{14}^5+x_{23}^5+x_{24}^5+x_{34}^5\\
&T_2=A^3(x_{12}^2+x_{13}^2+x_{14}^2)+B^3(x_{12}^2+x_{23}^2+x_{24}^2)+C^3(x_{13}^2+x_{23}^2+x_{34}^2)+D^3(x_{14}^2+x_{24}^2+x_{34}^2)\\
&T_3=A^2Bx_{12}+A^2Cx_{13}+A^2Dx_{14}+B^2Cx_{23}+B^2Dx_{24}+C^2Dx_{34}\\
&~~~~+AB^2x_{12}+AC^2x_{13}+AD^2x_{14}+BC^2x_{23}+BD^2x_{24}+CD^2x_{34}\\
&T_4=A^2(x_{12}^3+x_{13}^3+x_{14}^3)+B^2(x_{12}^3+x_{23}^3+x_{24}^3)+C^2(x_{13}^3+x_{23}^3+x_{34}^3)+D^2(x_{14}^3+x_{24}^3+x_{34}^3)\\
&T_5=A^2(x_{12}^2x_{13}+x_{12}^2x_{14}+x_{13}^2x_{14}+\text{(exp. sw.)})+B^2(x_{23}^2x_{12}+x_{24}^2x_{12}+x_{23}^2x_{24}+(\text{exp. sw.}))\\
&~~~~+C^2(x_{34}^2x_{13}+x_{34}^2x_{23}+x_{23}^2x_{13}+(\text{exp. sw.}))+D^2(x_{24}^2x_{14}+x_{34}^2x_{14}+x_{34}^2x_{24}+(\text{exp. sw.}))\\
&T_6=A(x_{12}^4+x_{13}^4+x_{14}^4)+B(x_{12}^4+x_{23}^4+x_{24}^4)+C(x_{13}^4+x_{23}^4+x_{34}^4)+D(x_{14}^4+x_{24}^4+x_{34}^4)\\
&T_7=A(x_{12}^3x_{13}+x_{12}^3x_{14}+x_{13}^3x_{14}+\text{(exp. sw.)})+B(x_{23}^3x_{12}+x_{24}^3x_{12}+x_{23}^3x_{24}+(\text{exp. sw.}))\\
&~~~~+C(x_{34}^3x_{13}+x_{34}^3x_{23}+x_{23}^3x_{13}+(\text{exp. sw.}))+D(x_{24}^3x_{14}+x_{34}^3x_{14}+x_{34}^3x_{24}+(\text{exp. sw.}))
\end{align*}

Note that while this minimal generating set contains 31 polynomials, only 8 different graph structures appear. These are depicted in Figure \eqref{fig:graphstruc}, where $n,m,k\in\mathbb{N}$ represent the different weights corresponding to each element in the minimal generating set.

\begin{figure}[H]
    \centering
    \includegraphics[width=0.7\linewidth]{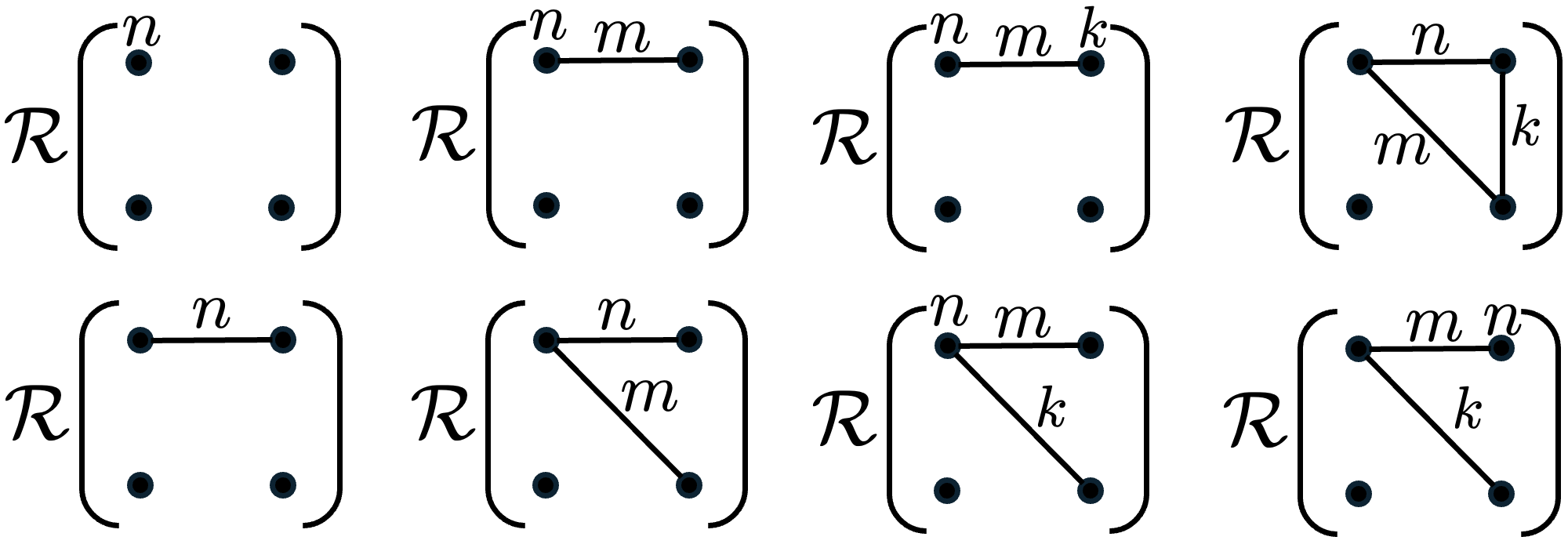}
    \caption{Types of invariant graph polynomials appearing in the minimal generating set for $\mathcal{I}^4$ given above.}
    \label{fig:graphstruc}
\end{figure}

\end{document}